\documentclass[conference]{IEEEtran}
\IEEEoverridecommandlockouts

\usepackage{cite}
\usepackage{amsmath,amssymb,amsfonts}
\usepackage{algorithmic}
\usepackage{algorithm}
\usepackage{graphicx}
\usepackage{textcomp}
\usepackage{xcolor}
\usepackage{subfigure}
\usepackage{bm}
\usepackage{booktabs} 
\newtheorem{Proposition}
{\textbf{Proposition}}
\newtheorem{Lemma}{\textbf{Lemma}}
\newtheorem{proof}{\textbf{Proof}}
\newtheorem{definition}{Definition}
\def\BibTeX{{\rm B\kern-.05em{\sc i\kern-.025em b}\kern-.08em   T\kern-.1667em\lower.7ex\hbox{E}\kern-.125emX}}

\makeatletter 
\newcommand{\linebreakand}{%
  \end{@IEEEauthorhalign}
  \mbox{}\par
  \mbox{}
  \begin{@IEEEauthorhalign}
}
\makeatother 

\begin{document}

\title{Guaranteeing and Explaining Stability across Heterogeneous Load Balancing using Calculus Network Dynamics

\thanks{The authors are with the Faculty of Engineering and Applied Sciences, Cranfield University, United Kingdom. The work is supported by EPSRC CHEDDAR: Communications Hub for Empowering Distributed clouD computing Applications and Research (EP/X040518/1) (EP/Y037421/1).
}
}

\author{
\IEEEauthorblockN{Mengbang Zou}
\IEEEauthorblockA{m.zou@cranfield.ac.uk} \\
\and
\IEEEauthorblockN{Yun Tang}
\IEEEauthorblockA{yun.tang@cranfield.ac.uk}
\and\IEEEauthorblockN{Adolfo Perrusqu\'ia}
\IEEEauthorblockA{adolfo.perrusquia-guzman@cranfield.ac.uk}
\and
\IEEEauthorblockN{Weisi Guo}
\IEEEauthorblockA{weisi.guo@cranfield.ac.uk} \\
}

\maketitle

\begin{abstract}

Load balancing between base stations (BSs) allows BS capacity to be efficiently utilised and avoid outages. Currently, data-driven mechanisms strive to balance inter-BS load and reduce unnecessary handovers. The challenge is that over a large number of BSs, networks observe an oscillatory effect of load evolution that causes high inter-BS messaging. Without a calculus function that integrates network topology to describe the evolution of load states, current data-driven algorithms cannot explain the oscillation phenomenon observed in load states, nor can they provide theoretical guarantees on the stability of the ideal synchronised state. Whilst we know load state oscillation is coupled with the load balancing process algorithms and the topology structure of inter-BS boundary relations, we do not have a theoretical framework to prove this and a pathway to improving load balancing algorithms. Here, we abstract generic and heterogeneous data-driven algorithms into a calculus dynamics space, so that we can establish the synchronization conditions for networked load balancing dynamics with any network topology. By incorporating what is known as “non-conservative error” and the eigenvalue spectrum of the networked dynamics, we can adjust the inter-BS load balancing mechanisms to achieve high efficiency and convergence guarantee, or to mitigate the oscillation when the synchronisation condition can not be satisfied.

\end{abstract}


\section{Introduction}
Load balancing is a critical function of national infrastructure, ensuring that network capacity is efficiently utilised and that spike demands are met by sharing the load. In telecommunication networks, this is primarily achieved by handing users off to neighbouring BSs with extra capacity, whilst minimising handover ping-pong effects due to the large number of mobile users and dynamics of traffic demand. In modern 5G and future 6G networks, the load balancing is not only in data traffic flows served under the BS wireless capacity but also edge computation requirements that may exceed BS edge computing capacity.
To meet the exponential increase of traffic data and edge computing demand, base stations are densely deployed, increasing the network’s capacity and throughput by reducing the distance to users \cite{li2018ultra}. 
However, unbalanced load distribution across the network can cause the under-utilization of resources where some BSs are heavily congested while others are idle, wasting spectrum and power resource. As a result, unbalanced load distribution can severely undermine the advantages of network with dense BSs. Load balancing aims to distribute user traffic evenly across base stations (BSs), preventing some BSs from being overloaded while others are underloaded, which is beneficial for the resource utilization and energy efficiency of wireless network as well as quality of services (QoS) \cite{shao2024access}. A lot of methods have been proposed to determine the amount of loads to share and the set of edge users to hand over. However, the stability and convergence of load evolution based on the corresponding load balancing algorithm is rarely studied. We observe that while the current load balancing algorithm effectively reduces the load difference among base stations, the system does not fully converge to an ideal state and instead continues to oscillate. Oscillation implies that handover happens frequently, but fail to drive the load balance toward a more convergent or optimal state. Even though the handover hysteresis parameter is usually set in algorithms to prevent frequent handover, the underlying cause of the oscillation, the oscillation bound (upper and lower bound), and whether the system can eventually converge to an ideal state remain unclear to us. Understanding the underlying causes behind these phenomena enables the design of appropriate strategies that lead to more uniform load balancing across the network and avoid unnecessary handovers.


\begin{figure*}[ht]
    \centering
    \resizebox*{17cm}{!}{\includegraphics{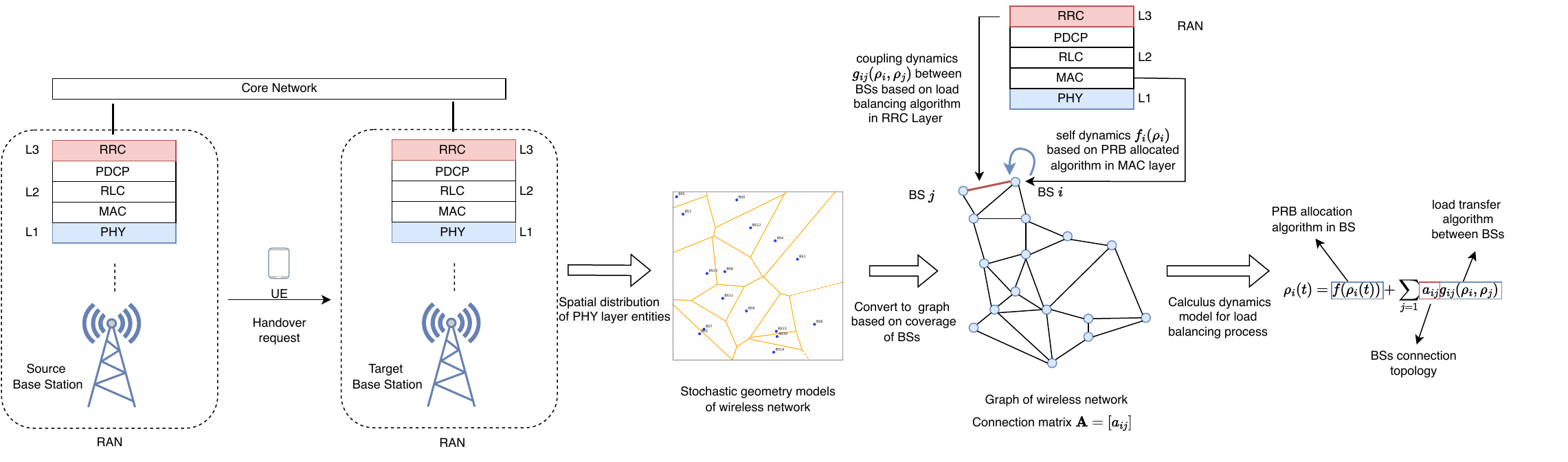}}
    \caption{Abstract the data-driven load balancing algorithms into a calculus networked dynamics model.}
    \label{fig: network_dynamics}
\end{figure*}

\begin{figure}[ht]
    \centering
    \resizebox*{8cm}{!}{\includegraphics{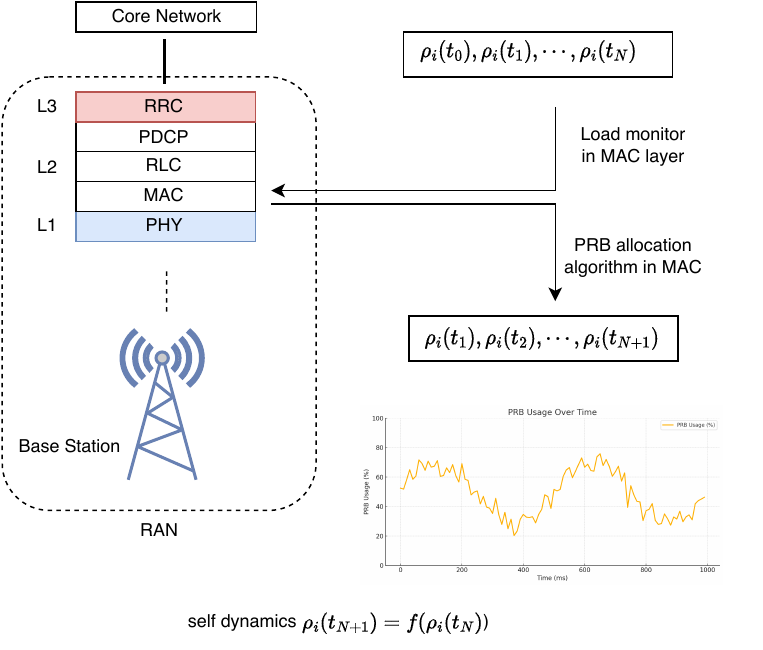}}
    \caption{Self dynamics based on the PRB allocation algorithm}
    \label{fig: self_dynamics}
\end{figure}

\begin{figure}[ht]
    \centering
    \resizebox*{8cm}{!}{\includegraphics{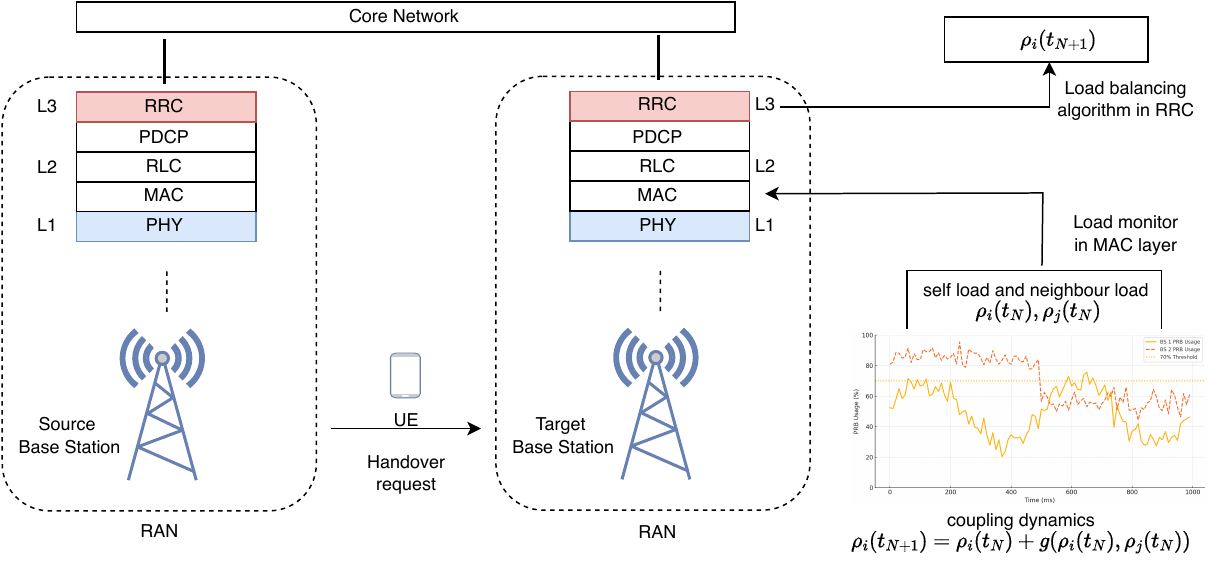}}
    \caption{Coupling dynamics based on load balancing algorithm.}
    \label{fig: coupling_dynamics}
\end{figure}

\subsection{Related work}
Current literature focuses on designing standard compliant uniform algorithms to address load unbalance problem. The load of BSs can be controlled through antenna beam tilting \cite{guo2013dynamic, fehske2013concurrent}, power control \cite{siomina2004optimization} to realise load balance across cellular network. More recently, a mobility load balancing (MLB) method has been proposed by tuning the value of cell individual offset (CIO) to adjust the logical cell boundaries of BSs \cite{kwan2010mobility, hasan2018adaptive, yang2012high}. By periodically monitoring the traffic load of BSs, overloaded BSs adjust the CIO to shrink the coverage and hand over edge users to the underloaded neighbouring BSs to alleviate their load. Some algorithms based on MLB have been proposed to optimise the CIO and determine the set of edge users to hand over. For example, a heuristic load balancing algorithm inspired by Kuramoto synchronisation model has been proposed in \cite{park2017mobility}. Deep reinforcement learning (DRL) methods have been successfully applied in load balancing problem \cite{wu2023reinforcement, xu2019load, alsuhli2021mobility, feriani2022multiobjective, sun2025proactive}. The size of the potential action space of DRL expands exponentially with the increase in the number of BSs, which poses significant challenges to scalability and efficient decision-making in load balancing. To overcome this problem, a decentralised DRL has been proposed in \cite{chang2022decentralized, alizadeh2024multi}.

While the aforementioned methods effectively distribute load evenly across the network, the convergence and stability of the resulting ideal load-balanced state remain unexplored. Stability in this context refers to the ability of all BSs to sustain an efficient and balanced load distribution, even in the presence of traffic fluctuations or external perturbations, which can be caused by the movement of users or the introduction of new users.  If the ideal load-balanced state is not stable, even minor perturbations can trigger constant handovers of edge users between base stations, resulting in degraded quality of service (QoS) and necessitating the use of complex hysteresis mechanisms to mitigate the instability. For an arbitrarily large network, instability can cause cascade handover (e.g., user handover in BS A causes unintended handovers in BS B, C, etc.) and hence energy efficiency problems, which have not been examined. Therefore, investigating the cascade stability of load balancing is a critical yet under explored area. Without a dynamics model for the load balancing process, it is difficult to analyse the evolution of load in each BSs in long term systematically.
Our previous work established a dynamical model for the load balancing process and derived an exact analytical condition for global stability, explicitly linking local load balancing dynamics with the underlying network topology \cite{moutsinas2019probabilistic, zou2022analysing, zou2025data}. It was based on the assumption that all BSs are in active mode and the load dynamics follows a homogeneous algorithm across the network. 

However, heterogeneous deployments, such as networks comprising macro, micro, and pico BSs \cite{wu2015energy}, \cite{wang2023base}, along with diverse control policies, have been widely adopted to achieve significant energy savings. Our previous work used regular patterns to examine cascade behaviour \cite{guo2013dynamic} and later proposed an analysis framework based on Gershgorin circles of the eigenvalues \cite{moutsinas2019probabilistic} are not valid in this heterogeneous scenario. To bridge this gap, in \cite{zou2024cascade}, we proposed a framework to analyse the stability of the synchronised state of traffic load balancing in heterogeneous networks, considering sleep mode with general load dynamic functions and random geographic networks.

\subsection{Research Gap}
According to the above stability analysis methods, the system can synchronise to the ideal state if the system satisfies some specific conditions. One of the condition is that the amount of decreased load in overloaded BS handing over user services equals to the increased load in the BS accepting this handover, which we call it conservative load transfer in this paper. With this condition, the stability function of the system has a symmetric structure and therefore can be simplified by eigen decomposition to relate the stability with eigenvalue spectrum of corresponding matrix. Without this symmetric structure, the stability analysis is much more complicated.

However, in reality, the handover is not always conservative. Instead, the non-conservative load transfers are commonly observed in the real world. This is because the Signal-to-Interference-plus-Noise Ratio (SINR) of a user may vary when the user's service is handed over from one BS to another. Therefore, the allocated resource, e.g. Physical Resource Block (PRB) to satisfy this user's requirement is different in BSs, which will cause the non-conservative load transfer. The existence of non-conservative load transfer will limit the applicability of our previously established theory in real-world scenarios. To bridge this gap, in this paper, we introduce the concepts, conservative load transfer and non-conservative load transfer. Based on this, we propose a theoretical framework to explain the reason why with the current load balancing algorithms, the load difference of BSs can effectively reduce at the beginning but not fully converge to the same state instead continues to oscillate. Also, the oscillation bound can be estimated according to this analysis framework. Based on the estimation of the oscillation bound, we can design an appropriate strategy to alleviate the bound to improve the performance of current load balancing algorithms.

\subsection{Novelty \& Contribution}
We observe the phenomenon that the current load balancing algorithm effectively reduces the load difference among BSs, however, the system does not always fully converge to an ideal state and instead continues to oscillate. The oscillation may cause the ping-pong effect without improving the performance of the system. To establish a theoretical framework to explain the potential reasons behind this phenomenon and provide a pathway to improve the load balancing algorithm

\begin{enumerate}

\item We establish a networked dynamic model of load evolution in the load balancing process, comprising self-dynamics within BSs, coupling dynamics across BSs, and the underlying network topology. The self-dynamics of load evolution is governed by the allocation algorithm in MAC layer, while the coupling dynamics is driven by the load balancing algorithm, which operates through logical network interfaces (e.g., F1-C/-U planes between radio units (RU)). The network topology is abstracted according to the stochastic spatial distribution of radio units (RUs) and the coverage relations in the wireless network. 
The algorithms are data-driven and can be anything reasonable, which is then projected onto a calculus (ODE) space for abstraction model.

\item Based on the dynamics model, we establish the stability criterion of the load balancing process for both continuous-time (if the load variation is sufficiently smooth and the balancing decisions are frequent made) and discrete-time dynamics models (the dynamics of load balancing are typically modeled in discrete time, as state updates and handover performs at discrete intervals in load balancing algorithms).

\item With the data-driven load balancing algorithms, the state of load evolution may oscillate instead of synchronising to the same state. To explain the potential reason behind this phenomenon, with the calculus dynamics model for the load balancing algorithm, we found that the existence of non-conservative load transfer causes the oscillation. The condition for the existence of oscillation bound has been established as well as the estimation of this bound if the non-conservative load transfer exists and the stability can not be guaranteed.

\item We found that the bound of the oscillation is decided by the non-conservation error, eigenvalue spectral and the amount of transfer load. Based on this, we can design appropriate strategy to alleviate the bound to improve the performance of current load balancing algorithms. 
    
\end{enumerate}

\section{System Setup}
\subsection{Model \& Assumptions}
Consider the area is covered by $N$ BSs. The signal-to-noise-plus-interference ratio (SINR) of user $u$ in BS $i$ at time $t$ is given by
\begin{equation}
    {\rm SINR}_{u,i}^t = \frac{P_i(t)G_{u,i}(t)}{N_0+\sum_{j \ne i}P_j(t)G_{u,j}(t)},
\end{equation}
where $P_i(t)$ is the transmission power of BS $i$ at time $t$, $G_{u,i}$ is the channel power gain between user $u$ and BS $i$ at time $t$, $N_0$ is the noise power. For user $u$ the maximum transmission rate from BS $i$ over one physical resource block (PRB) is 
\begin{equation}
    R_{u,i}^t=\Theta \log_2(1+ {\rm SINR}^t_{u, i}),
\end{equation}
where $\Theta$ is the spectrum bandwidth of one PRB.  Assume that each user has a constant bit rate requirement $\hat{R}_u(t)$ at time $t$. The number of required PRBs from BS $i$ to meet the demand $\hat{R}_u(t)$ is given by
\begin{equation}
    B_{u,i}(t) \ge \frac{\hat{R}_u^t}{R_{u, i}^t},
\end{equation}
The load of BS $i$ is 
\begin{equation}
    \rho_i(t) = \frac{\sum_{u=1}^{U_i(t)}B_{u, i}(t)}{B_i},
\end{equation}
where $B_i$ denotes the total number of PRBs in BS $i$. 

The handover process aims to transfer the user's service from its serving BS to a neighbouring BS. The handover process is triggered according to event A3, defined by 3GPP, 
\begin{equation}\label{equ: handover}
    M_i+\theta_{ij}>Hys + M_j,
\end{equation}
where $M_i$ and $M_j$ are the user measured values of reference signal received power (RSRP) from BS $i$ and $j$, $\theta_{ij}$ is the BS individual offset value of BS $i$ concerning BS $j$. $Hys$ is a hysteresis parameter defined for all BSs. The CIO is an offset that can be applied to alter the handover decision, which affects the radio service coverage of BS. Traffic load can be handed over to neighbour BSs with relatively low traffic load by adjusting the value of CIO, if the traffic load is high in the BS. The load balancing algorithm is based on the relative load levels compared to the neighbouring BSs. The objective of load balancing is 
\begin{equation}
    {\rm min} \sqrt{\sum_{i=1}^N(\rho_i(t)-\bar{\rho})^2},
\end{equation}
where $\bar{\rho}$ is the ideal state of load. The ideal situation is that $\rho_1(t) = \rho_2(t) = \cdots =\rho_N(t) = \bar{\rho}$.

\subsection{Calculus dynamics model for load balancing}
The state of the load in each BS is normally decided by the PRB scheduling algorithm within the BS, the handover algorithm among BSs and the last state of the related BSs (shown in Fig.~\ref{fig: network_dynamics}). Therefore, the evolution of the load in each BS can be written in a general form as 

\begin{equation}\label{equ: 7}
    \rho_i(t+1) = \rho_i(t) + f_i(\rho_i(t)) + \sum_{j=1, j\ne i}a_{ij}g_{ij}(\rho_i(t), \rho_j(t)), 
\end{equation}
where $\rho_i(t+1)$ is the next state of load in BS $i$,  $f_i(\rho_i(t))$ describes the load evolution of BS $i$ based on the PRB scheduling algorithm in MAC layer of BS $i$. For example, Round Robin (RR) \cite{mehlfuhrer2011vienna}, proportional fair scheduler \cite{liu2011asymptotic}, etc. $g_{ij}(\rho_i(t),\rho_j(t))$ describes the load sharing dynamics between BSs according to the handover algorithms. Generally, handover only happens when one BS is overloaded and its neighbouring nodes is underloaded. The amount of offload is based on the difference between $\rho_i(t)$ and $\rho_j(t)$. Therefore, $g_{ij}(\rho_i, \rho_j)$ can be further rewritten as $g_{ij}(\rho_i(t)-\rho_j(t))$. For example, in the handover algorithm in \cite{park2017mobility}, the amount of load can be shared to neighbouring BSs is $sin(\rho_j(t)-\rho_i(t))$, then the next state of load is $\rho_i(t+1) = \rho_i(t)+sin(\rho_j(t)-\rho_i(t))$ if only consider one sharing neighbour BS. $a_{ij}$ is the element of the connection matrix $\bm{A}$ which can describe any topology structure of connections between BSs in wireless networks. The network topology is abstracted according to the stochastic spatial distribution of radio
units (RUs) in the PHY layer and the coverage relations in the wireless
network (shown in Fig.~\ref{fig: network_dynamics}). If the coverage radius $r_i, r_j$ of RUs in BS $i$ and $j$ satisfy that $r_i + r_j > d_{ij}$, where $d_{ij}$ is the distance between RU $i$ and $j$, then BS $i$ and $j$ can share load with each other and $a_{ij}=1$, otherwise $a_{ij}=0$.

In the next section, the methodology to analyse load dynamics behaviour of BSs will be introduced.  Continuous-time (if the load variation is sufficiently smooth and the balancing decisions are frequently made, the discrete dynamics in equation~(\ref{equ: 7}) can be rewritten as $\frac{d\rho_i(t)}{dt}$) and discrete-time dynamics models (the dynamics of load balancing are typically modeled in discrete time, as state updates and handover performs at discrete intervals in load balancing algorithms) are both considered.

\section{Methodology}
\subsection{Continuous dynamics of load balance}
\subsubsection{Homogeneous Load Balancing Dynamics}
First, we consider a simple case to illustrate the dynamics behaviour of the system. Assume that $\rho_i(t) \in \mathbb{R}$ is continuous in $[0, 1]$ and there exist continuous functions $f(\cdot)$ and $g(\cdot)$ can describe the load balancing algorithm. Assume that the PRB scheduling algorithm is the same in each BS. The handover algorithms among BSs are the same. The dynamics of $\rho_i(t)$ is  
\begin{equation}\label{equ: homo dynamics}
    \frac{d\rho_i(t) }{dt}= f(\rho_i(t)) + \sum_{j=1,j \ne i} a_{ij}g(\rho_i(t)-\rho_j(t)).
\end{equation}

To synchronize to the ideal state $\rho_1(t)=\rho_2(t)=\cdots \rho_N(t) = \bar{\rho}$, it requires equation~(\ref{equ: homo dynamics}) satisfies two conditions: 

\noindent 1) $\frac{d \rho_i(t)}
{dt}\mid_{\rho_i(t)=\bar{\rho}}=0$; 

\noindent 2) $\frac{d\dot{\rho}_i(t)}{d\rho_i}\mid_{\rho_i(t)=\bar{\rho}}<0$. 

The first condition indicates the existence of the synchronized equilibrium and the second one indicates the stability of this equilibrium. Generally, the handover is triggered only if $\rho_i(t)$ is overloaded, assume BSs are in active mode, e.g. $\rho_i(t)>\bar{\rho}$ and the neighbour BS $j$ is under-loaded e.g. $\rho_j(t)<\bar{\rho}$. 

If $\rho_i(t)=\rho_j(t)=\bar{\rho}$, then $g(0)=0$. Therefore, to satisfy condition 1), it requires that $f(\bar{\rho})=0$. To analyze the stability of the equilibrium, equation~(\ref{equ: homo dynamics}) is expanded around equilibrium $\bar{\rho}$ by Taylor series as $f(\rho_i(t))=f(\bar{\rho})+f'(\bar{\rho})(\rho_i(t)-\bar{\rho})+ o(\rho_i(t)-\bar{\rho})$, where $o$ represents higher order infinitesimal. Here we only consider the first order expansion, so 
\begin{equation}
    f(\rho_i(t)) \approx \bar{\rho}+ f'(\bar{\rho})(\rho_i(t)-\bar{\rho})
\end{equation}
Similarly, we expand $g(\rho_j(t)-\rho_i(t))$ around $(\bar{\rho}, \bar{\rho})$. Let $z=\rho_j(t)-\rho_i(t)$.

\begin{equation}
    g(\rho_i(t)-\rho_j(t)) \approx g(0)+\frac{\partial g}{\partial \rho_i}(\rho_i(t)-\bar{\rho})+\frac{\partial g}{\partial \rho_j}(\rho_j(t)-\bar{\rho})
\end{equation}

\[
    \frac{\partial g}{\partial \rho_i} = \frac{d g}{d z} \frac{\partial z}{\partial \rho_i} = g'(\rho_i(t)-\rho_j(t))
\]

\[
    \frac{\partial g}{\partial \rho_j} = \frac{d g}{d z} \frac{\partial z}{\partial \rho_j} = -g'(\rho_i(t)-\rho_j(t))
\]

\[
    g(\rho_i(t)-\rho_j(t)) \approx g'(0)(\rho_i(t)-\rho_j(t))
\]

Write equation~(\ref{equ: homo dynamics}) in a matrix form based on Taylor series around the equilibrium, we can get
\begin{equation}
    \dot{\bm{\rho}} = f'(\bar{\rho})\bm{I}(\bm{\rho}-\bar{\bm{\rho}}) + g'(0)\bm{L}\bm{\rho},
\end{equation}
where $\bm{I}$ is the identical matrix, $\bm{\rho} \in \mathbb{R}^{N \times 1} =[\rho_1, \rho_2, \cdots, \rho_N]^{\top}$. $\bm{L} = \bm{D}-\bm{A}$ is a Laplacian matrix, where $\bm{D}$ is a diagonal matrix with $\bm{D} \in \mathbb{R}^{N \times N} = {\rm diag}\{D_1, D_2, ... D_N\}$, $D_i = \sum_{j \ne i}a_{ij}$. $\bar{\bm{\rho}} = \bar{\rho}\bm{1}$, where $\bm{1} \in \mathbb{R}^{N \times 1}= [1, 1, \cdots, 1]^{\top}$. Since $\bm{L}$ is a Laplacian matrix, it always can decompose as $\bm{L}=\bm{S}^{-1}\bm{\Lambda}\bm{S}$, where $\bm{\Lambda}$ is a diagonal matrix, $\bm{\Lambda}={\rm diag}\{\lambda_1, \lambda_2, \cdots \lambda_N\}$. $\lambda_i$ is eigenvalue of $\bm{L}$ and $0 = \lambda_1 \le \lambda_2 \le \cdots \le \lambda_N$.
\begin{equation}\label{equ: matrix form}
    \dot{\bm{\rho}} = \bm{S}^{-1}(f'(\bar{\rho})\bm{I}+g'(0)\bm{\Lambda})\bm{S}\bm{\rho}-f'(\bar{\rho})\bar{\bm{\rho}}
\end{equation}
For simplicity, let $\tilde{\bm{\Lambda}} = f'(\bar{\rho})\bm{I}+g'(0)\bm{\Lambda}$. The stability of equation~(\ref{equ: matrix form}) is decided by the distribution of eigenvalues of $\tilde{\bm{\Lambda}}$, $\tilde{\lambda}_i$. If ${\rm Re}(\tilde{\lambda})_i<0$, then the system is asymptotic stable around the equilibrium. Since load sharing is from overloaded BS to the under-loaded one, $g(\rho_i(t)-\rho_j(t))<0$, $\rho_i(t)-\rho_j(t)>0$; $g(\rho_j(t)-\rho_i(t))=0$, $\rho_j(t)-\rho_i(t)=0$; $g(\rho_i(t)-\rho_j(t))>0$, $\rho_i(t)-\rho_j(t)<0$. Therefore, $g'(0)<0$. For $f(\cdot)$, if $\rho_i(t)<\bar{\rho}$, the BS can allocate more PRB to users to improve QoS. If $\rho_i(t)>\bar{\rho}$ the BS will allocate less PRB to users to alleviate BS's load. Therefore, $f'(\bar{\rho})<0$. If the load balance algorithm satisfies these two conditions i.e. $f'(\bar{\rho})<0$, $g'(0)<0$, then the system can synchronise to the ideal state $\bar{\rho}$ and maintain stable.

\subsubsection{Heterogeneous Load Balancing Dynamics with conservative load transfer}
We now consider a more general scenario in which the function $f(\cdot)$ is specific to each BS, and $g(\cdot)$ likewise varies among BSs. In this case, the load dynamics of each BS is 
\begin{equation}\label{equ: coupling_dynamics}
    \dot{\rho}_i(t) = f_i(\rho_i(t))+\sum_{j=1, j\ne i}^Na_{ij}g_{ij}(\rho_i(t)-\rho_j(t)),
\end{equation}
where $f_i(\rho_i(t))$ is the self-load dynamics of a BS related to the PRB, $g_{ij}(\rho_j(t)-\rho_i(t))$ is the offloading dynamics between BSs related to the handover process as discussed before.

The desirable equilibrium for maximum service efficiency is $\rho_i(t)= \bar{\rho}$. For any dynamics, $f_i(\bar{\rho})=0, g_{ij}(0)=0$. If all BSs achieve the same state $\rho_1(t), \rho_2(t), \cdots, \rho_N(t) =\bar{\rho}$, then we say the whole system synchronises to a desirable state. Obviously, $\rho_i(t)=\bar{\rho}$ is a solution of equation~(\ref{equ: coupling_dynamics}). The problem is whether $\rho_i(t) = \bar{\rho}$ is stable with the small perturbation $\Delta \rho_i(t)$. Assume that $\Delta \rho_i(t) = \rho_i(t)-\bar{\rho}$. $f_i(\rho_i(t))$ and $g_{ij}(\rho_i(t),\rho_j(t))$ can be expressed as 
\[f_i(\rho_i(t))=f_i(\bar{\rho})+f'_i(\bar{\rho})\Delta \rho_i(t)\]
 
\[\begin{aligned}
g_{ij}(\rho_i(t),\rho_j(t))=g_{ij}(\bar{\rho},\bar{\rho})+\frac{\partial g_{ij}(\bar{\rho},\bar{\rho})}{\partial \rho_i(t)}\Delta \rho_i(t) \\
+\frac{\partial g_{ij}(\bar{\rho},\bar{\rho})}{\partial \rho_j(t)}\Delta \rho_j(t)
\end{aligned}\] by the first-order Taylor series.

Since $g_{ij}(\rho_i(t),\rho_j(t))$ always satisfies $g_{ij}(\rho_i(t), \rho_j(t))=0$ with $\rho_i(t)=\rho_j(t)$, $g_{ij}(\rho_i(t), \rho_j(t))$ has the term $\rho_i(t)-\rho_j(t)$ and it is natural to assume that 
\[g_{ij}(\rho_i(t), \rho_j(t))=(\rho_i(t)-\rho_j(t))h_{ij}(\rho_i(t), \rho_j(t)).\] 
Then $\frac{\partial g_{ij}(\bar{\rho},\bar{\rho})}{\partial \rho_i(t)}=h_{ij}(\bar{\rho},\bar{\rho}), \frac{\partial g_{ij}(\bar{\rho},\bar{\rho})}{\partial \rho_j(t)}=-h_{ij}(\bar{\rho},\bar{\rho})$. We have $\frac{\partial g_{ij}(\bar{\rho},\bar{\rho})}{\partial \rho_i(t)}=-\frac{\partial g_{ij}(\bar{\rho},\bar{\rho})}{\partial \rho_j(t)}$. Expand equation~(\ref{equ: coupling_dynamics}) around the equilibrium $\bar{\rho}$ as
\begin{equation} \label{equ: load_taylor}
    \dot{\rho}_i(t) = f_i'(\bar{\rho})\Delta \rho_i(t) + \sum_{j=1, j\ne i} a_{ij}h_{ij}(\bar{\rho},\bar{\rho})(\Delta \rho_i(t) -\Delta \rho_j(t))
\end{equation}

Rewrite equation~(\ref{equ: load_taylor}) as a matrix form, then we have 
\begin{equation}
    \dot{\bm{\rho}} = \bm{D}_f \Delta \bm{l}+ (\bm{K}-\bm{Q})\Delta \bm{l}
\end{equation}
where $\bm{D}_f={\rm diag}\{f_1'(\bar{\rho}), f_2'(\bar{\rho}), \cdots, f_N'(\bar{\rho})\}$, ${\rm diag} \{ \}$ represents the diagonal matrix. $\bm{Q} = \bm{A}\odot \bm{H}$, where $\odot$ is a Hadamard product.
\begin{equation}
    \bm{H} =  \begin{bmatrix}
    0 & h_{12}(\bar{\rho}, \bar{\rho}) & \cdots & h_{1N}(\bar{\rho}, \bar{\rho}) \\
    h_{21}(\bar{\rho}, \bar{\rho}) & 0  & \cdots & h_{2N}(\bar{\rho}, \bar{\rho}) \\
    \vdots & \vdots & \vdots & \vdots\\
    h_{N1}(\bar{\rho}, \bar{\rho}) & h_{N2}(\bar{\rho}, \bar{\rho}) &  \cdots & 0
    \end{bmatrix}
\end{equation}
$\bm{K}$ is a diagonal matrix, where $k_{ii} = \sum_{j=1}^{N}q_{ij}$, $q_{ij}$ are the elements of $\bm{Q}$.
According to Lyapunov's second method for stability, if there exists a Lyapunov function $V(\bm{x})=\bm{x}^{\top}\bm{P}\bm{x} > 0$ ($V(\bm{x})=0$ only if $\bm{x}=0$), where $\bm{P}$ is positive definite, and $\dot{\bm{V}}(\bm{x})<0$, then the system is stable at the equilibrium. Let $\delta{\bm{l}}=\bm{x}$,then $\dot{\bm{V}}(\bm{x}) = \dot{\bm{x}}^{\top}\bm{P}\bm{x}+\bm{x}^{\top}\bm{P}\dot{\bm{x}}$ can be expanded as 
\begin{equation}
  \dot{\bm{V}}(\bm{x})=(\bm{x}^{\top}\bm{D}_f^{\top}+\bm{x}^{\top}(\bm{K}^{\top}-\bm{Q}^{\top}))\bm{P}\bm{x}+\bm{x}^{\top}\bm{P}(\bm{D}_f\bm{x}+(\bm{K}-\bm{Q})\bm{x})
\end{equation}
Let $\bm{P}=\bm{I}$, where $\bm{I}$ is identity matrix. Then
\begin{equation}
    \dot{\bm{V}}(\bm{x})=2(\bm{x}^{\top}\bm{D}_f\bm{x}+\bm{x}^{\top}(\bm{K}-\bm{Q})\bm{x}).
\end{equation}

If $h_{ij}(\bar{\rho},\bar{\rho})<0$, $k_{ii}<0$. Let $\bm{M}=\bm{K}-\bm{Q}$ and  $\lambda_i$ be the eigenvalue of matrix $\bm{M}$. According to the Gershgorin circle theorem \cite{varga2011gervsgorin},

\begin{equation}\label{equ: gershgorin}
    |\lambda_i - M_{ii}| \le \sum_{j=1, j\ne i}|M_{ij}|.
\end{equation}
Since $M_{ii}=\sum_{j=1, j\ne i}M_{ij}$, $\lambda_i \le 0$. 

\begin{definition}
    If BS $i$ handover users' service to BS $j$, then the amount of load decreased by $i$ equals the amount increased by $j$: $\Delta \rho_i + \Delta\rho_j = 0$, then this kind of handover is defined as conservative load transfer. Otherwise the handover is non-conservative load transfer.
\end{definition}

With conservative load transfer, $g_{ij} = g_{ji}$. 
Then $h_{ij}=h_{ji}$ and $\bm{M}$ is symmetric. According to lemma~(\ref{Lem2}), we can get $\dot{\bm{V}}(x)<0$. Therefore, with conservative load transfer, the system need to satisfy following conditions to achieve a synchronised ideal load balance state:

\noindent 1) $f_i(\bar{\rho})=0$, $g_{ij}(\bar{\rho}, \bar{\rho})=0.$

\noindent 2) $f_i'(\bar{\rho})<0$, $h_{ij}=h_{ji}<0$.

This means that to keep the stability of the ideal synchronised state, BSs can only handover edge users to BSs with less load. Besides, within the BS, if BS is overloaded, then it should allocate fewer PRBs to users. 

\subsubsection{Heterogeneous Load Balancing Dynamics with non-conservative load transfer}
However, the conservative load transfer is an ideal situation and cannot always be guaranteed. Here, we briefly explain why non-conservative transfer always happen. If a user $v$ is handed over from BS $j$ to BS $i$, the perturbation of load in BS $i$ is $\Delta \rho_i=-\frac{B_{v,i}(t)}{B_i(t)}$ and the perturbation in BS $j$ is $\Delta \rho_j = \frac{B_{v,j}(t)}{B_j(t)}$. Normally, the perturbation of load in BS $i$ and $j$ caused by user $v$ move from $j$ to $i$ is different since ${\rm SINR}_{v,i}^t$ and ${\rm SINR}_{v,j}^t$ may be different as well as $B_i(t)$ and $B_{j}(t)$. 

In this case, $\bm{M}$ is not symmetric, which means even eigenvalues of $\bm{M}$ $\lambda_i \le 0$ can not guarantee $\bm{x}^{\top}\bm{M}\bm{x} \le 0$.

\begin{Lemma}\label{Lem1}
    For a square matrix $\bm{M}$, all its eigenvalues $\lambda_i \le 0$. If $\bm{S}=\frac{\bm{M}^{\top}+\bm{M}}{2}$ satisfies that for any $\bm{x}$, $\bm{x}^{\top}\bm{S}\bm{x} \le 0$, then $\bm{x}^{\top}\bm{M}\bm{x} \le 0$.
\end{Lemma}

\begin{proof}
    Assume $\bm{B}=\frac{\bm{M}-\bm{M}^{\top}}{2}$. Then $\bm{B}^{\top}=\frac{\bm{M}-\bm{M}^{\top}}{2}=-\bm{B}$. We have $\bm{x}^{\top}\bm{B}\bm{x}=(\bm{x}^{\top}\bm{B}\bm{x})^{\top}=\bm{x}^{\top}\bm{B}^{\top}\bm{x}$, which means that $\bm{x}^{\top}(\bm{B}-\bm{B}^{\top})\bm{x}=2\bm{x}^{\top}\bm{B}\bm{x}=0$. If $\bm{x}^{\top}\bm{S}\bm{x} \le 0$, $\bm{x}^{\top}(\bm{S}+\bm{B})\bm{x} \le 0$. $\bm{S}+\bm{B}=\bm{M}$, therefore, $\bm{x}^{\top}\bm{M}\bm{x} \le 0$.
\end{proof}

\begin{Lemma}\label{Lem2}
    For any real symmetric matrix $\bm{S}$, if all eigenvalues $\lambda_i \le 0$, then for any $\bm{x}, \bm{x}^{\top}\bm{S}\bm{x}\le 0$ 
\end{Lemma}

\begin{proof}
    For any real symmetric matrix $\bm{S}$, it is always can be decomposed as $\bm{S} = \bm{U}^{\top}\bm{\Lambda}_s\bm{U}$, where $\bm{\Lambda}_s$ is a diagonal matrix $\bm{\Lambda}_s = {\rm diag}\{\lambda_1, \lambda_2, \cdots, \lambda_N\}$, where $\lambda_i$ is the eigenvalue of $\bm{S}$. $\bm{x}^{\top}\bm{S}\bm{x}=(\bm{Ux})^{\top}\bm{\Lambda}_s\bm{Ux} \le 0$
\end{proof}

$\bm{S}=\frac{\bm{M}^{\top}+\bm{M}}{2}=\bm{S}^{\top}$ is a symmetric matrix. According to the Proposition~(\ref{Lem2}), if all eigenvalues of $\bm{S}$, $\lambda_i \le 0$, for any $\bm{x}$, $\bm{x}^{\top}\bm{S}\bm{x} \le 0$.  $\bm{S}$ has the form
\begin{equation} \label{equ: matrixS}
    \begin{bmatrix}
    \sum_j a_{1j} h_{1j} & -\frac{a_{12}(h_{12}+h_{21})}{2} & \cdots & -\frac{a_{1N}(h_{1N}+h_{N1})}{2} \\
    -\frac{a_{21}(h_{12}+h_{21})}{2} & 0  & \cdots & -\frac{a_{2N}(h_{2N}+h_{N2})}{2} \\
    \vdots & \vdots & \vdots & \vdots\\
    -\frac{a_{N1}(h_{1N}+h_{N1})}{2} &-\frac{a_{N2}(h_{N2}+h_{2N})}{2}  &  \cdots & \sum_j a_{Nj} h_{Nj}
    \end{bmatrix}
\end{equation}
For simplification, in the equation~(\ref{equ: matrixS}), we use $h_{ij}$ to represent $h_{ij}(\bar{\rho}, \bar{\rho})$ here. 

Based on the Gershgorin circle theorem in equation~(\ref{equ: gershgorin}), eigenvalues $\lambda_i$ of matrix $\bm{S}$ satisfy $|\lambda_i - S_{ii}| \le \displaystyle \sum_{j=1, j\ne i}|S_{ij}|.$ $S_{ii}- \displaystyle \sum_{j=1, j\ne i}S_{ij}=1/2(-\displaystyle \sum_{j=1, j\ne i}a_{ij}h_{ij}(\bar{\rho},\bar{\rho})+\displaystyle \sum_{j=1, j\ne i}a_{ji}h_{ji}(\bar{\rho}, \bar{\rho}))$. We can obtain if $\displaystyle \sum_{j=1, j\ne i}|S_{ij}| \ge \sum_{j=1, j\ne i}|S_{ji}|$, all eigenvalues of matrix $\bm{S}$, $\lambda_i\le 0$. According to lemma~(\ref{Lem1}), for any $\bm{x}$, $\bm{x}^{\top}\bm{M}\bm{x}\le 0$. Since $f_i(\bar{\rho})<0$, $\dot{\bm{V}}(x)<0$ for any $\bm{x} \ne 0$. Therefore, if $|\displaystyle \sum_{j=1, j\ne i}a_{ij}h_{ij}(\bar{\rho}, \bar{\rho})| \ge |\displaystyle \sum_{j=1, j\ne i}a_{ji}h_{ji}(\bar{\rho}, \bar{\rho})|$, then the system can maintain the stability around the ideal equilibrium $\bar{\rho}$. This implies that for any BS, the load sharing rate to the other BSs is larger than the rate of load receiving from other BSs, then the BS can maintain the ideal stable equilibrium.

Now consider the self dynamics and coupling dynamics $\bm{D}_f, \bm{K}$ and $\bm{Q}$ together, and let $\bm{M}=\bm{D}_f+\bm{K}-\bm{Q}$. $\bm{S}=\frac{\bm{M}+\bm{M}^{\top}}{2}$. $S_{ii}= f'_i(\bar{\rho})+\displaystyle \sum_{j=1, j\ne i}a_{ij}h_{ij}(\bar{\rho}, \bar{\rho})$ and $S_{ij}=1/2(\displaystyle \sum_{j=1, j\ne i}a_{ij}h_{ij}(\bar{\rho}, \bar{\rho})+\displaystyle \sum_{j=1, j\ne i}a_{ji}h_{ji}(\bar{\rho}, \bar{\rho}))$. Similarly to the previous analysis, we can conclude if $| 2f'_i(\bar{\rho})+\displaystyle \sum_{j=1, j\ne i}a_{ij}h_{ij}(\bar{\rho}, \bar{\rho})| \ge |\displaystyle \sum_{j=1, j\ne i}a_{ji}h_{ji}(\bar{\rho}, \bar{\rho})|$. This is a more relaxed condition to maintain stability. Therefore, the stability condition is:

\noindent 1) $f_i(\bar{\rho})=0$, $g_{ij}(\bar{\rho}, \bar{\rho})=0$

\noindent 2) $| 2f'_i(\bar{\rho})+\displaystyle \sum_{j=1, j\ne i}a_{ij}h_{ij}(\bar{\rho}, \bar{\rho})| \ge |\displaystyle \sum_{j=1, j\ne i}a_{ji}h_{ji}(\bar{\rho}, \bar{\rho})|$

This indicates that if the rate to decrease BS's load by self-adjusting and offloading is larger than the load increasing rate by accepting handover from others, then it can maintain stability. This means that the BS has the capability to tolerate the upcoming load shared by neighbours to maintain the ideal state.

\subsection{Discrete dynamics of load balance}
\subsubsection{Conservative Load Transfer}
Indeed, most load balancing algorithms primarily focus on load distribution among base stations (BSs), while rarely addressing resource adjustment within a base station (intra-BS). Therefore, the dynamics of load is mainly decided by the load sharing dynamics between BSs. Besides, load balancing algorithms usually update load step by step, so the dynamics model can be established in discrete time. The load of BS $i$ is updated by 
\begin{equation} \label{equ: identical load dynamics}
    \rho_i(t+1) = \rho_i(t) + \epsilon\sum_{j=1, \neq i}a_{ij}(\rho_j(t)-\rho_i(t))
\end{equation}

\begin{Proposition}
    The system can achieve the same state $\rho_1(t)=\rho_2(t)=\cdots \rho_N(t)$, $t\to +\infty$ if $\epsilon < \frac{2}{\lambda_{\rm max}}$. 
\end{Proposition}

\begin{proof}
    It is easy to check $\rho_i(t) = \rho_2(t)=\cdots \rho_N(t) = \bar{\rho}(t)$ is the solution of Equation~(\ref{equ: identical load dynamics}). Equation~(\ref{equ: identical load dynamics}) can be rewritten in a matrix form as $\bm{\rho}(t+1)=(\bm{I}-\epsilon \bm{L})\bm{\rho}(t)$, where $\bm{L} = \bm{D}-\bm{A}$ is the Laplacian matrix corresponding to $\bm{A}$ and $\bm{D}$ is the diagonal matrix where $D_{ii}=\sum_{j\ne i}A_{ij}$. $\bm{L}$ can be decomposed as $\bm{L}=\bm{S}^{-1} \bm{\Lambda} \bm{S}$, where $\bm{\Lambda} = {\rm diag}\{\lambda_1, \lambda_2, \cdots, \lambda_N\}$. In cellular networks like LTE and 5G, resource management tasks—including PRB allocation, CIO tuning, and user handovers—are performed periodically according to the frame structure. It is natural to analyze the load $\rho_i(t)$ in discrete time as $\rho_i(k)$ and $\bm{\rho}(k+1)=(\bm{I}-\epsilon \bm{L})\bm{\rho}(k) = (\bm{I}-\epsilon \bm{L})^k\bm{\rho}(0)$. $(\bm{I}-\epsilon \bm{L})^k=(\bm{S}^{-1}(\bm{I}-\epsilon\bm{\Lambda})\bm{S})^k=(\bm{S}^{-1}(\bm{I}-\epsilon\bm{\Lambda})^k\bm{S})$. The system converges to the balance state requires that $|1-\epsilon \lambda_i|_{i\ge2} < 1$. Eigenvalues of Laplacian matrix have the property that $0=\lambda_1 \le \lambda_2 \le \cdots \lambda_N$. Therefore, this requires that $0<\epsilon < \frac{2}{\lambda_{\rm max}}$.
\end{proof}

$|1-\epsilon\lambda_i|$ decides the convergence speed of the system. Let's define
\[\eta = \arg\max_{i\ge2}|1-\epsilon\lambda_i|\]

\begin{figure}[t]
    \centering
    \resizebox*{8.5cm}{!}{\includegraphics{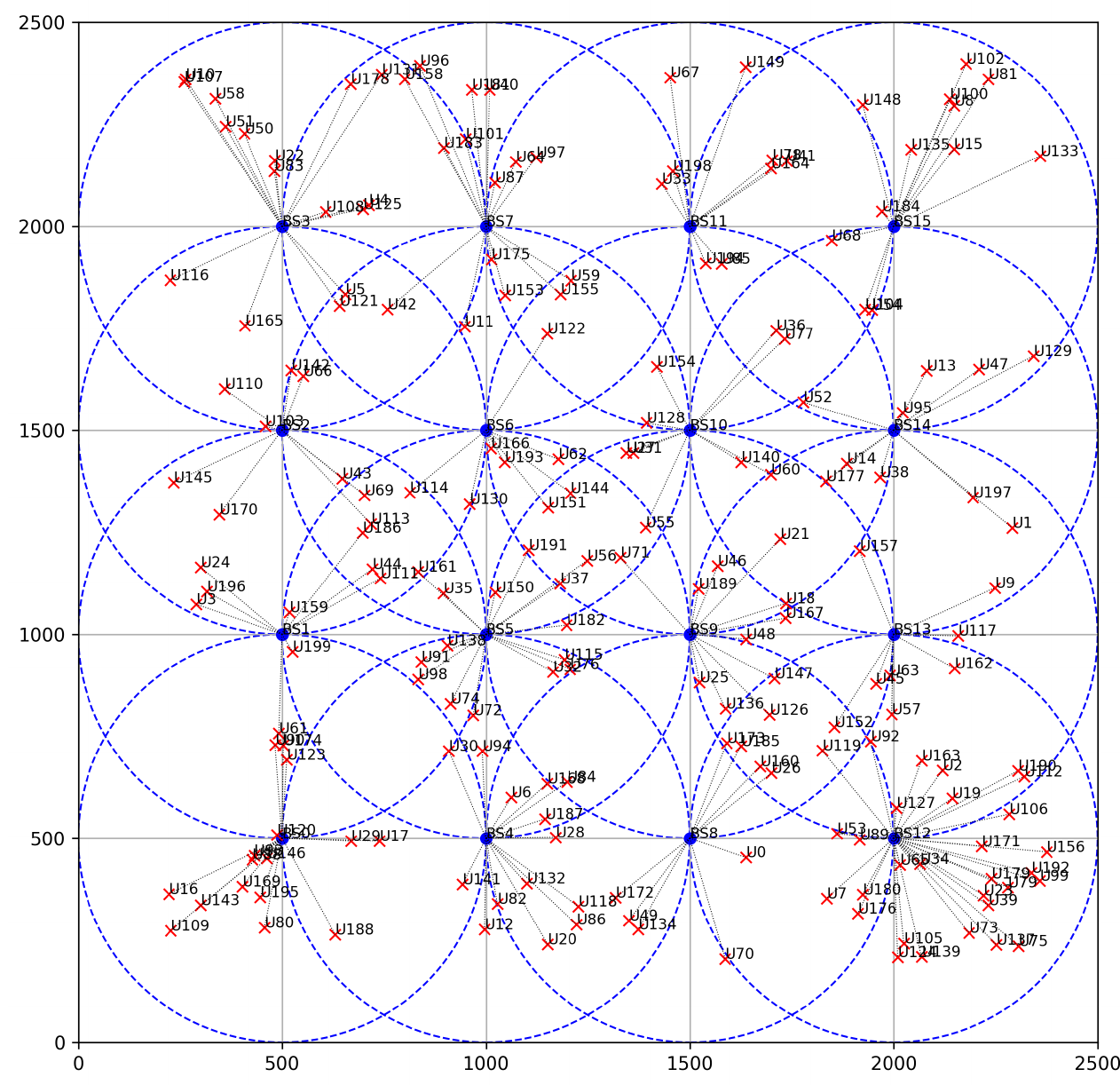}}
    \caption{Simulation of BSs and users distribution. Blue nodes represent BSs and red nodes are users.}
    \label{fig: BSs distribution}
\end{figure}


\subsubsection{Non-conservative load transfer}
If a user $v$ is handed over from BS $j$ to BS $i$, the perturbation of load in BS $i$ is $\Delta \rho_i=\frac{B_{v,i}(t)}{B_i(t)}$ and the perturbation in BS $j$ is $\Delta \rho_j = \frac{B_{v,j}(t)}{B_j(t)}$. Normally, the perturbation of load in BS $i$ and $j$ caused by user $v$ move from $j$ to $i$ is different since ${\rm SINR}_{v,i}^t$ and ${\rm SINR}_{v,j}^t$ may be different as well as $B_i(t)$ and $B_{j}(t)$. So, we can not simply assume the load share dynamics is $\epsilon(\rho_j-\rho_i)$. The dynamics of the load is 

\begin{equation}\label{equ: non-conservative}
    \rho_i(k+1) = \rho_i(k)+\sum_{j=1, j \ne i}\epsilon_{ij}(k)a_{ij}(\rho_j(k)-\rho_i(k)) 
\end{equation}
$\epsilon_{ij}(k)$ maybe not equal to $\epsilon_{ji}(k)$, $\epsilon_{ij}(k) < 1$.

To analyse how the non-conservative load transfer affects the convergence of the load balance evolution, rewrite equation~(\ref{equ: non-conservative}) as
\begin{equation} \label{equ: unbalance}
    \rho_i(k+1) = \rho_i(k)+\sum_{j=1, j \ne i}\tilde{\epsilon}_{ij}(k)A_{ij}(\rho_j(k)-\rho_i(k)) + \omega_i(k), 
\end{equation}
where $\tilde{\epsilon}_{ij}(k)=\frac{\epsilon_{ij}(k)+\epsilon_{ji}(k)}{2}$, 
\[\omega_i(k)=\sum_{j=1, j \ne i}\frac{\epsilon_{ij}(k)-\epsilon_{ji}(k)}{2}A_{ij}(\rho_j(k)-\rho_i(k))\]

Rewirte equation~(\ref{equ: unbalance}) in a matrix form, we can obtain
\begin{equation}
    \bm{\rho}(k+1) = (\bm{I}-\hat{\bm{L}}(k))\bm{\rho}(k)+\bm{\omega}(k),
\end{equation}
where $\hat{L}_{ij} = -\tilde{\epsilon}_{ij}(k), \forall i\ne j$ and $\hat{L}_{ii}=\sum_{j=1}\tilde{\epsilon}_{ij}(k)$. $\hat{\bm{L}}(k)$ is symmetric matrix. The error between $\rho_i(k)$ and $\bar{\rho}(k)$ is defined as $e_{i}(k) = \rho_i(k)-\bar{\rho}(k)$. $\bm{e}(k) \in \mathbb{R}^{N*1}$ is the  vector of $e_i(k)$. 
\[ \bm{e} = \bm{\rho}(k)- \bar{\bm{\rho}}(k) \]
, where $\bar{\bm{\rho}}(k) = \bar{\rho}\bm{1}$, $\bm{1} \in \mathbb{R}^{N*1}$ is a 1 vector, $\bar{\rho} = \frac{1}{N}\bm{1}^{\top}\bm{\rho}$.
\[\bm{e}(k+1) = \bm{\rho}(k+1)-\bar{\bm{\rho}}(k+1)\]
\[\bm{e}(k+1) =(\bm{I}-\hat{\bm{L}(k)})\bm{\rho}(k) +\bm{\omega}(k) - (\frac{1}{N}\bm{1}^{\top}((\bm{I}-\hat{\bm{L}}(k))\bm{\rho}(k)+\bm{\omega}(k)))\bm{1}\]

\[\bar{\rho}(k+1) = (\frac{1}{N}\bm{1}^{\top}((\bm{I}-\hat{\bm{L}}(k))\bm{\rho}(k)+\bm{\omega}(k)))
\]

\[\bm{1}^{\top}(\bm{I}-\hat{\bm{L}}(k)) = \bm{1}^{\top}\] since $\bm{1}^{\top}\hat{\bm{L}}(k)=0$. Then
\[\bar{\rho}(k+1) = \frac{1}{N}\bm{1}^{\top}(\bm{\rho}(k)+\bm{\omega}(k)) = \bar{\rho}(k)+\bar{\omega}(k),\]
where $\bar{\omega}(k) = \frac{1}{N}\bm{1}^{\top}\bm{\omega}(k)$.

So we have
\[\bm{e}(k+1) = (\bm{I}-\hat{\bm{L}}(k))\bm{\rho}(k)+\bm{\omega}(k) - (\bar{\bm{\rho}}(k)+\bar{\bm{\omega}}(k)) \],
$\bar{\bm{\omega}}(k)=\bar{\omega}\bm{1}$

Since $\bm{e}(k) = \bm{\rho}(k) - \bar{\bm{\rho}}(k)$, $\hat{\bm{L}}(k) \bar{\bm{\rho}}(k)=0$, we have
\[(\bm{I}-\hat{\bm{L}}(k))(\bm{\rho}(k) - \bar{\bm{\rho}}(k)) = (\bm{I}-\hat{\bm{L}}(k))\bm{\rho}(k) - \bar{\bm{\rho}}(k)\]
Therefore
\[\bm{e}(k+1)=(\bm{I}-\hat{\bm{L}}(k))(\bm{\rho}(k)-\bar{\bm{\rho}}(k))+\bm{\omega}(k) - \bar{\bm{\omega}}(k) \]
The evolution of error is 
\begin{equation}
    \bm{e}(k+1)=(\bm{I}-\hat{\bm{L}}(k))\bm{e}(k)+\bm{\omega}(k) - \bar{\bm{\omega}}(k).
\end{equation}

Here we define the Lyapunov function $V(\bm{e}(k)):= \bm{e}(k)^{\top}\bm{e}(k)$.


\begin{equation}\label{equ: error}
    \begin{split}
        &V(\bm{e}(k+1)) = \\
        &((\bm{I}-\hat{\bm{L}}(k))\bm{e}(k)+ \bm{\omega}(k) - \bar{\bm{\omega}}(k))^{\top} \\
& ((\bm{I}-\hat{\bm{L}}(k))\bm{e}(k)+\bm{\omega}(k) - \bar{\bm{\omega}}(k))
    \end{split}
\end{equation}

If the absolute value of perturbation of load in $i$ and $j$ caused by the handover is the same, then $\epsilon_{ij}=\epsilon_{ji}$,  $\omega_i(k)=0$. $V(\bm{e}(k+1)) = \bm{e}(k)^{\top}(\bm{I}-\hat{\bm{L}}(k))^2\bm{e}(k)$. $V(\bm{e}(k+1)) \le {\rm max}(|\gamma_i|)^2V(k) $. When ${\rm max}(|\gamma_i|)<1$, then $V(\bm{e}(k+1)) < V(\bm{e}(k))$, where $\gamma_i$ is the eigenvalue of $\bm{I}-\hat{\bm{L}}(k)$.



Let $\hat{\gamma}(k) = {\rm max}|\gamma_i(k)|$, where $\gamma_i(k)$ is eigenvalue of $\bm{I}-\hat{\bm{L}}(k)$. Let $\alpha(k) = || \bm{\omega}(k)-\bar{\bm{\omega}}(k) ||_2$. According to Cauchy-Schwarz Inequality 

$((\bm{I}-\hat{\bm{L}}(k))\bm{e}(k))^{\top}(\bm{\omega}(k)-\bar{\bm{\omega}}(k)) \le \\ ||(\bm{I}-\hat{\bm{L}}(k))\bm{e}(k))^{\top}||_2||\bm{\omega}(k)-\bar{\bm{\omega}}(k)||_2$

$||(\bm{I}-\hat{\bm{L}}(k))\bm{e}(k))^{\top}||_2 \le \hat{\gamma}(k)||V(\bm{e}(k))||_2$. 

Then we have
\[V(\bm{e}(k+1)) \le \hat{\gamma}(k)^2V(\bm{e}(k)) + \alpha(k)^2 + 2\hat{\gamma}(k)\alpha(k)||V(\bm{e}(k))||_2 \]
\[V(\bm{e}(k+1))\le (\hat{\gamma}(k)||V(\bm{e}(k))||_2+\alpha(k))^2\]

Assume that there exists a bound that $$V(\bm{e}(k)) \le \tilde{V}$$,  which requires $V(\bm{e}(k+1)) \le V(\bm{e}(k))$ whenever $V(\bm{e}(k))\ge \tilde{V}$. 

\[(\tilde{\gamma}(k)||V(\bm{e}(k))||_2+\alpha(k))^2 \le V(\bm{e}(k))\]

\begin{equation}\label{equ: bound}
    \tilde{V} = (\frac{\alpha(k)}{1-\tilde{\gamma}(k)})^2
\end{equation}
From equation~(\ref{equ: bound}), we can see that the bound of oscillation is determined by the load transfer error $\alpha(k)$ and the spectral radius of $\tilde{\gamma}(k)$. If the load transfer is conservative that $\alpha(k)=0$, then $V(\bm{e}(k+1))$ always smaller than $V(\bm{e}(k))$, which indicates that the error is monotonically decreasing and will converge to $0$ finally. Based on equation~(\ref{equ: bound}), we can design the algorithm to reduce the bound by reducing the load transfer error. The key point here to reduce the load transfer error is to construct the preferred users for each BS based on the load perturbation.

\begin{algorithm} 
\caption{Load balance Algorithm}    \begin{algorithmic}\label{algorithm: 1}
    \FOR{$\rho_i(k)$ $\in$ $\bm{\rho}$}
     \IF{$\rho_i(k) > \rho_{\rm{th}} \land \eta_i(k)<0$}
     \STATE Handover in BS $i$ is triggered;
     \FOR{BS $j$ $\in$ $\mathcal{N}_i$}
     \IF{$\rho_j(k)<\rho_{\rm th} \land \eta_j(k)>0$}
     \STATE BS $i$ send handover request to BS $j$;
     \STATE BS $i$ construct the user handover request list $\mathcal{V}_i(k)$; 
        \FOR{user $u$ $\in$ $\mathcal{V}_i(k)$}
            \STATE calculate the difference $|\rho_{u,i}(k)-\rho_{u,j}(k)|$;
        \ENDFOR
     \ENDIF
      \STATE Rank the user in $\mathcal{V}_i(k)$ based on load difference$|\rho_{u,i}(k)-\rho_{u,j}(k)|$ in ascending order as $\mathcal{V}^*_i(k)$;
     \ENDFOR
     \ENDIF
    \ENDFOR

    \WHILE{$\rho_i(k)>\rho_{\rm th} \land \rho_j(k) <\rho_{\rm th}$ or $k<k_{\rm th}$}
    \FOR{Each BS $i$ has a user handover request $\mathcal{V}^*_i(k)$} 
    
    \IF{$\rho_{u,i}(k) \le \eta_i(k)$ $\land$ $\mathcal{V}_i^*(k)$ is not empty}
    \STATE BS $i$ accept the handover request of user $u$ from BS $j$ according to the rank of $\mathcal{V}_i^*(k)$;
    \STATE Remove $u$ from all handover request $\mathcal{V}^*(k)$;
    \STATE Update information of all BS stations $\{\eta_i(k), \rho_i(k), \mathcal{V}_i^*(k) \}$
    \ENDIF
    \ENDFOR
    \ENDWHILE
    \FOR{Each pair of neighbouring BSs}
    \STATE $\theta_{ij} = {\rm maxmize}({\rm SINR}_{u,j}+ Hys - {\rm SINR}_{u,i})$ if user u was handed over from BS $j$ to $i$.
    \ENDFOR
\end{algorithmic}
\end{algorithm}

At the beginning, if BS $i$ is overloaded, $\rho_i(k)>\rho_{\rm th}$, and its load larger than the neighbouring BSs', $\eta_i(k)=\sum_{j=1, j \ne i}\epsilon_{ij}(k)A_{ij}(\rho_j(k)-\rho_i(k))<0$. For the neighbouring BSs of BS $i$, if $\rho_j(k)<\rho_{\rm th}$ and has the accommodation for load, $\eta_i >0$, then BS $i$ will send a handover request to BS $j$. The list of handover request of BS $i$ from neighbouring BSs is $\mathcal{V}_i$ is ranked in ascending order according to the difference of load $|\rho_{u,i}(k)-\rho_{u, j}(k)|$ where $\rho_{u,i}(k) = \frac{B_{u,i}(k)}{B_i}$, $\rho_{u, j}=\frac{B_{u,j}(k)}{B_j}$. $B_{u,i}(k)$ is the PRB requirement of user $u$ if $u$ is handed over to the BS $i$ and $B_{u,j}(k)$ is the PRB requirement of user $u$  BS $j$. The steps of the algorithm is shown in Algorithm~\ref{algorithm: 1}. 

In the above analysis, we established dynamic models for different scenarios, including both discrete and continuous forms, with or without considering the PRB scheduling mechanism within the base station. Non-conservative load transfer can prevent the system from synchronising to an ideal state. However, the presence of the PRB scheduling mechanism within the base station can help mitigate this issue, allowing the system to still converge to a desirable equilibrium.

\begin{figure}[ht]
  \centering

  \subfigure[]{
    \includegraphics[width=0.8\linewidth]{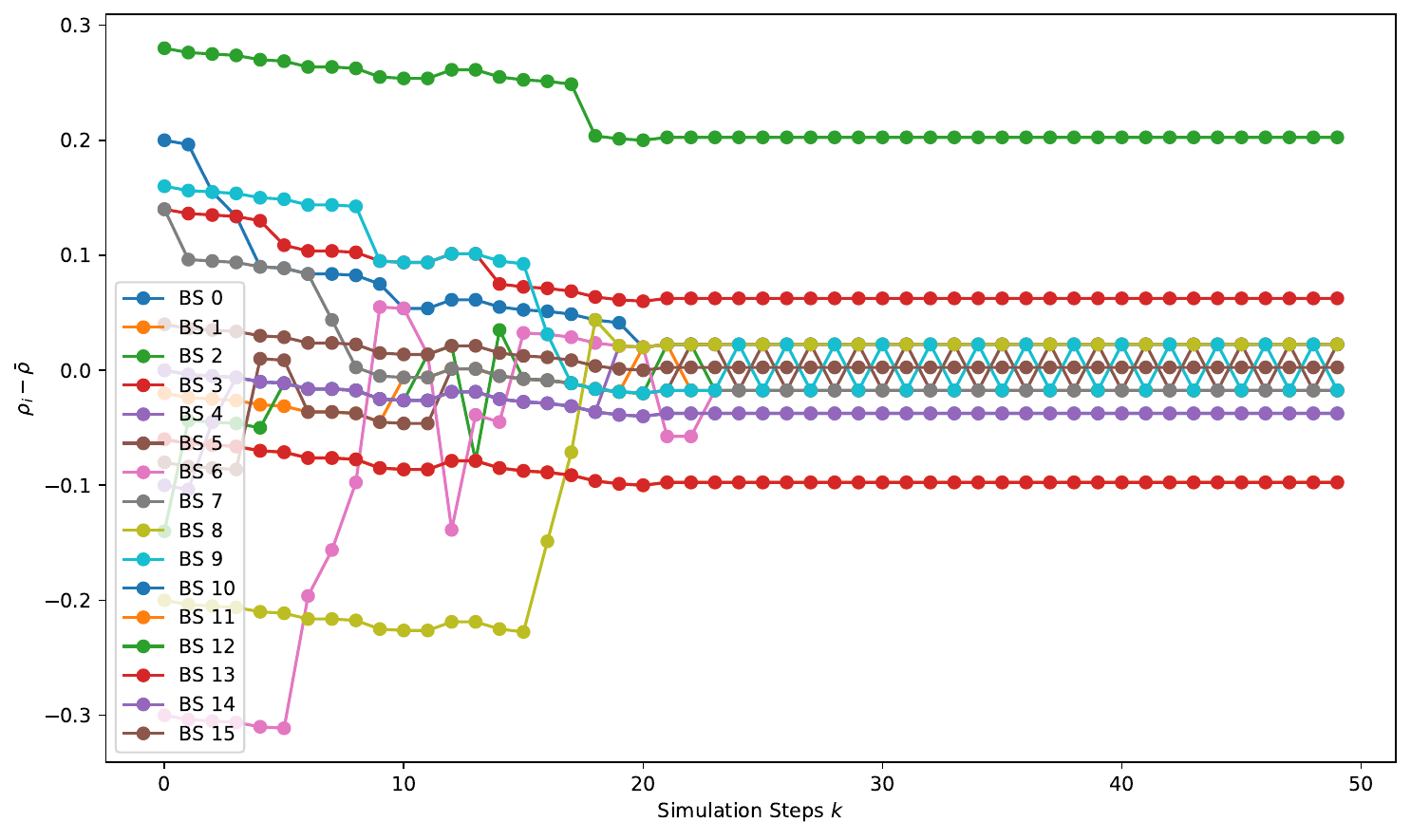}
  }
  
  \subfigure[]{
    \includegraphics[width=0.8\linewidth]{
    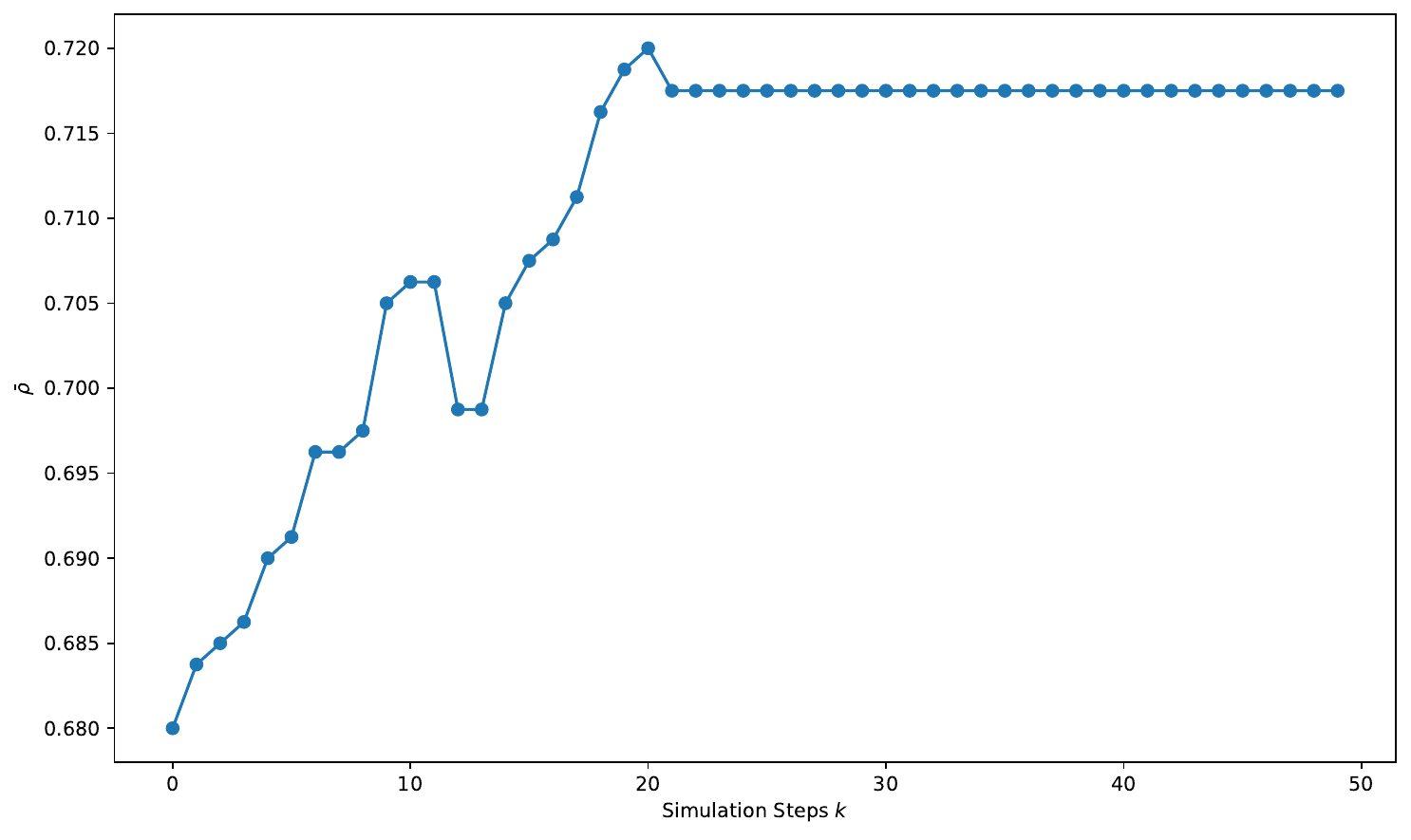}
  }

  \caption{Load balance based on load difference between BSs.}
  \label{fig: average load}
\end{figure}

\begin{table}[h]
\centering
\caption{}\label{tab: 1}
\begin{tabular}{ccccc} 
\toprule
Users & Average load & Algorithm 1 & Algorithm 2 & Algorithm 3 \\
\midrule
200 & 0.2712 &0.0842 & 0.0926 & 0.0886 \\
250 & 0.3387 &0.0917 & 0.0987 & 0.0953 \\
300 & 0.4087 &0.0727 & 0.0783 & 0.0778 \\
350 & 0.4750  &0.0607 & 0.0615 & 0.0569 \\
400 & 0.5412 &0.0449 & 0.0460 & 0.0494 \\
500 & 0.6799 &0.0574 & 0.0589 & 0.0650 \\
\bottomrule
\end{tabular}
\end{table}

\begin{table}[h]
\centering
\caption{}\label{tab: 2}
\begin{tabular}{ccccc} 
\toprule
Users & Average load & Algorithm 1 & Algorithm 2 & Algorithm 3 \\
\midrule
200 & 0.2712 &0.0842 & 0.0926 & 0.0886 \\
250 & 0.3387 &0.0917 & 0.0987 & 0.0953 \\
300 & 0.4087 &0.0619 & 0.0699 & 0.0667 \\
350 & 0.4750  &0.0541 & 0.0618 & 0.0569 \\
400 & 0.5412 &0.0449 & 0.0399 & 0.0494 \\
500 & 0.6799 &0.0574 & 0.0615 & 0.0589 \\
\bottomrule
\end{tabular}
\end{table}

\begin{table}[h]
\centering
\caption{}\label{tab: 3}
\begin{tabular}{ccccc} 
\toprule
Users & Average load & Algorithm 1 & Algorithm 2 & Algorithm 3 \\
\midrule
200 & 0.2712 &0.0818 & 0.0926 & 0.0886 \\
250 & 0.3387 &0.0917 & 0.0925 & 0.0953 \\
300 & 0.4087 &0.0595 & 0.0659 & 0.0799 \\
350 & 0.4750 &0.0533 & 0.0597 & 0.0556 \\
400 & 0.5412 &0.0413 & 0.0424 & 0.0442 \\
500 & 0.6799 &0.0537 & 0.1048 & 0.0589 \\
\bottomrule
\end{tabular}
\end{table}

\begin{figure*}[ht]
  \centering

  \subfigure[]{
    \includegraphics[width=0.3\linewidth]{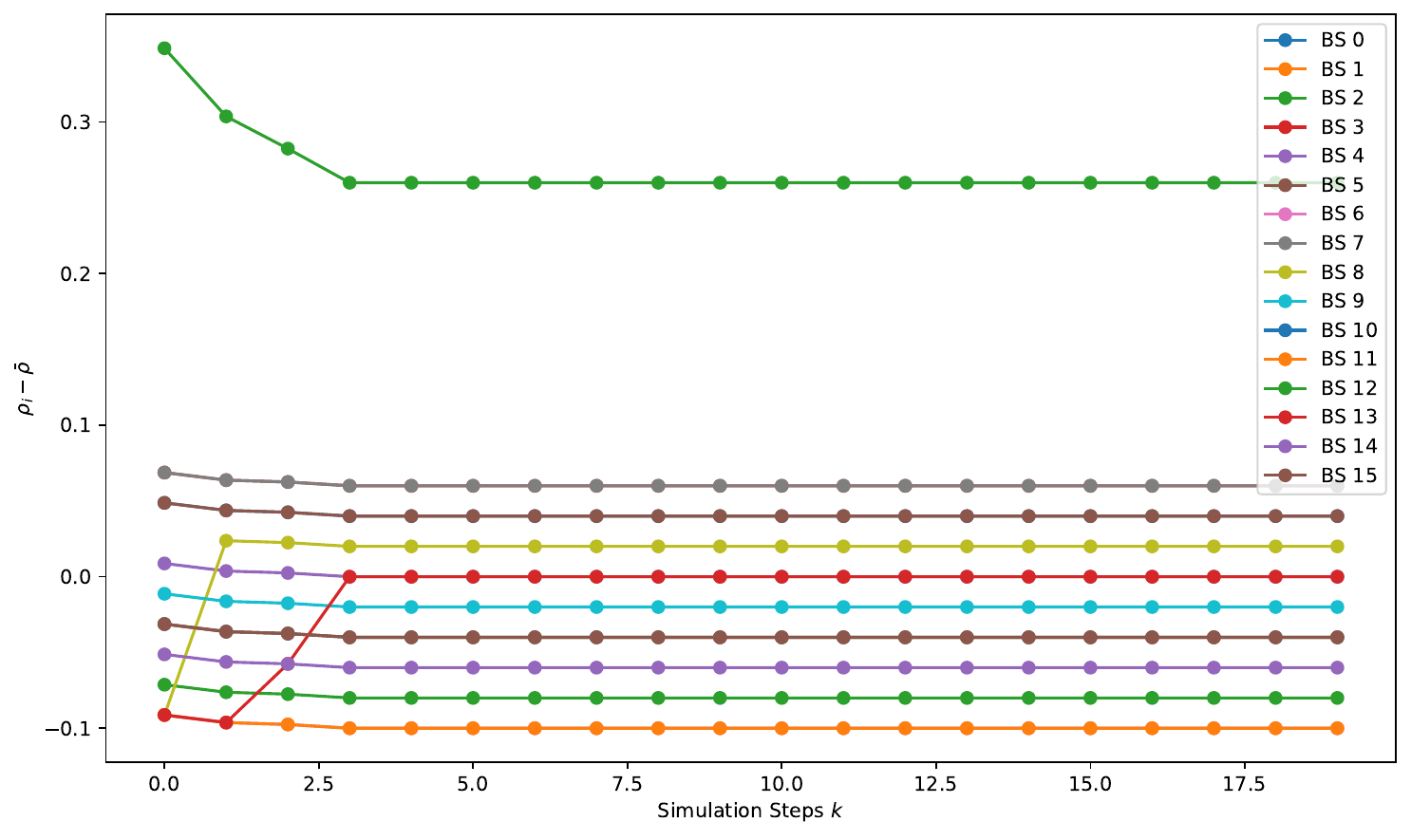}
  }
  \subfigure[]{
    \includegraphics[width=0.3\linewidth]{
    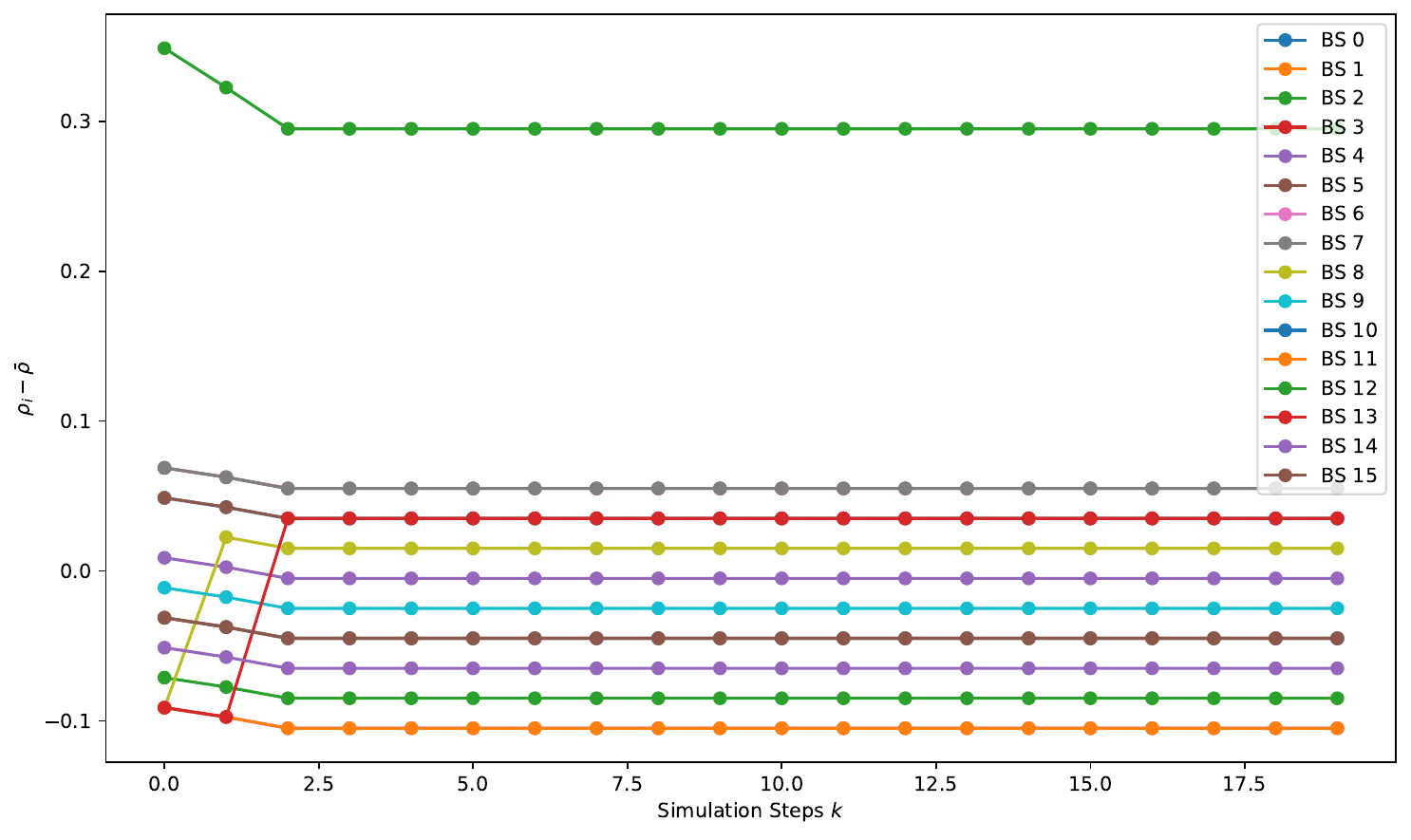}
  }
  \subfigure[]{
    \includegraphics[width=0.3\linewidth]{
    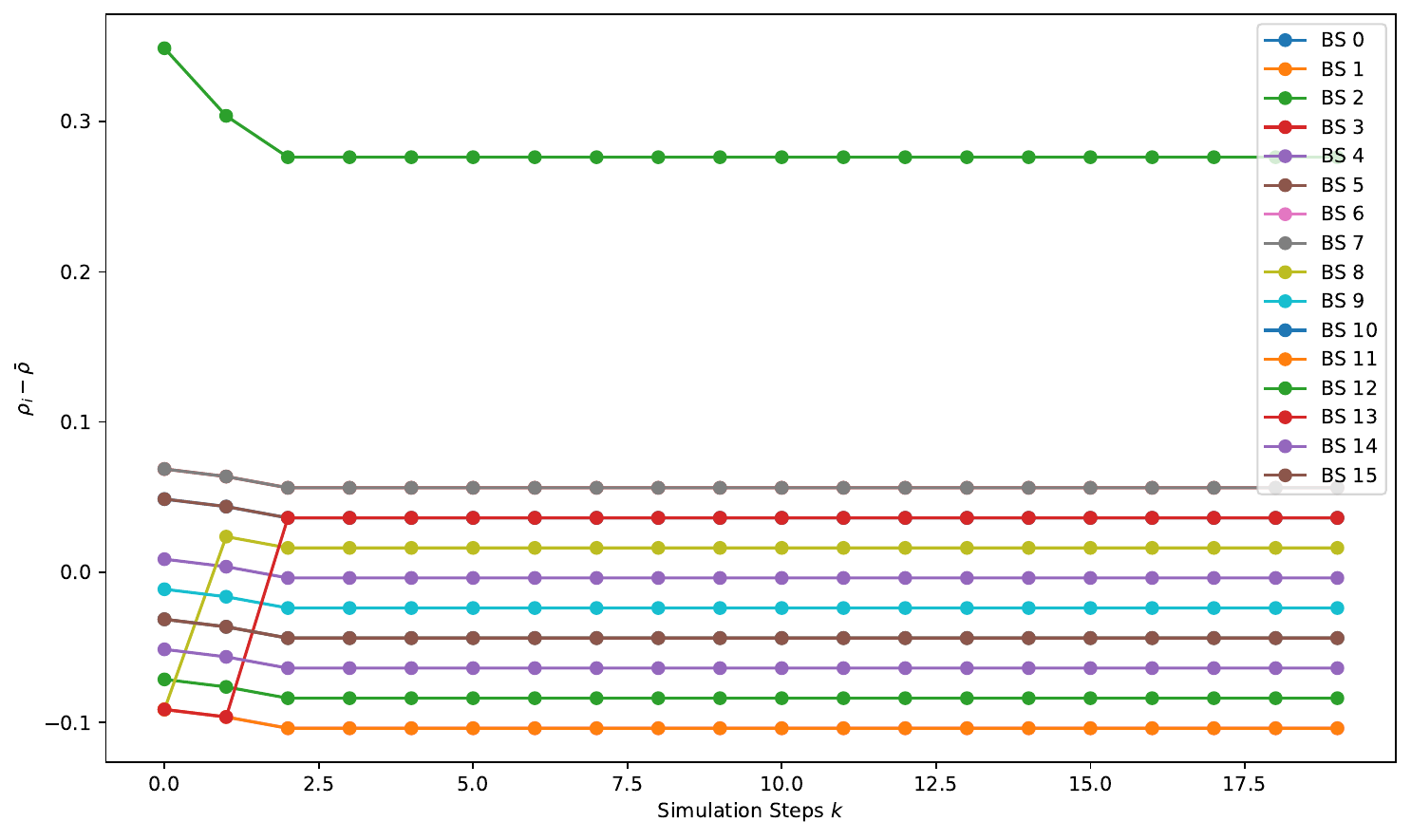}
  }

\subfigure[]{
    \includegraphics[width=0.3\linewidth]{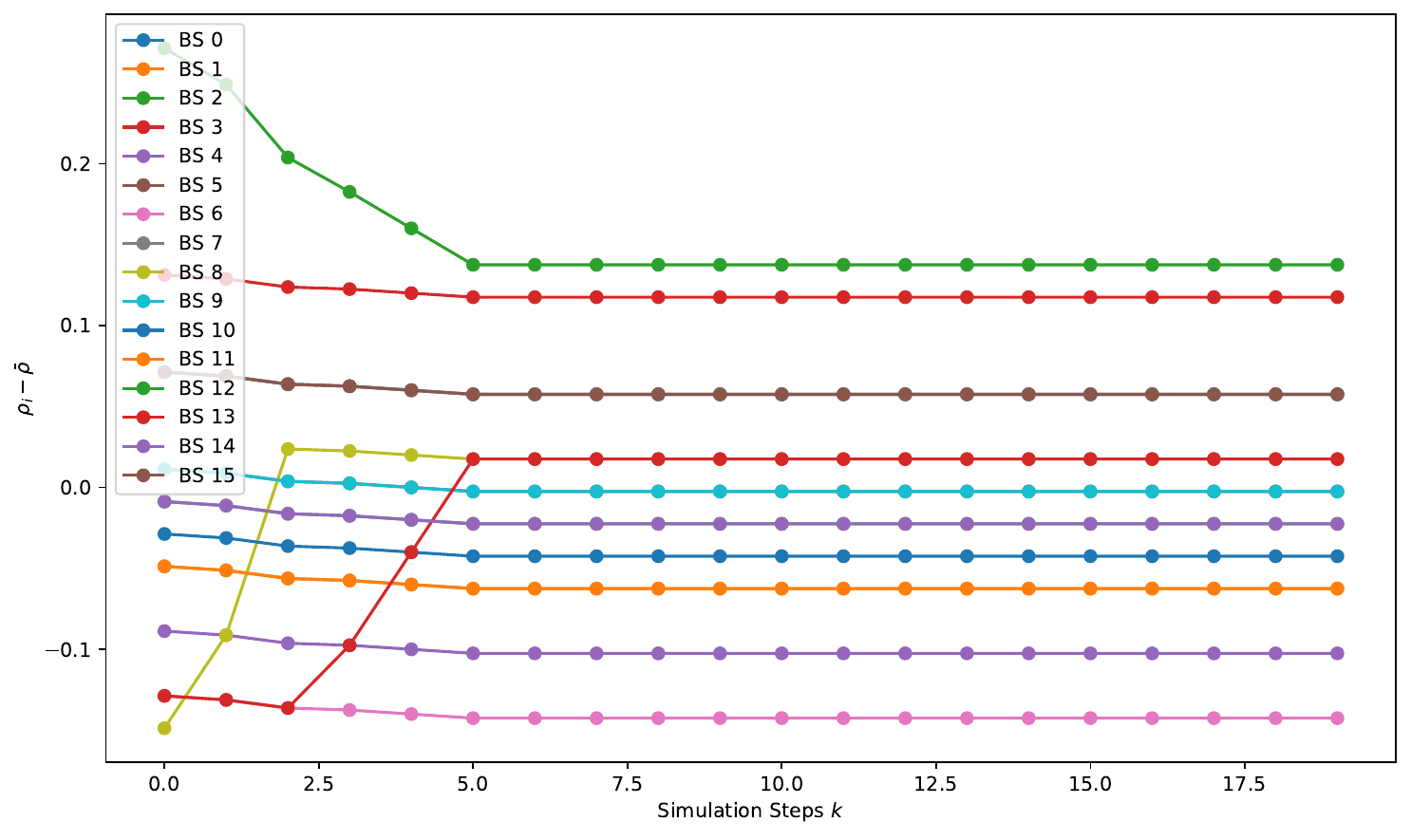}
  }
  \subfigure[]{
    \includegraphics[width=0.3\linewidth]{
    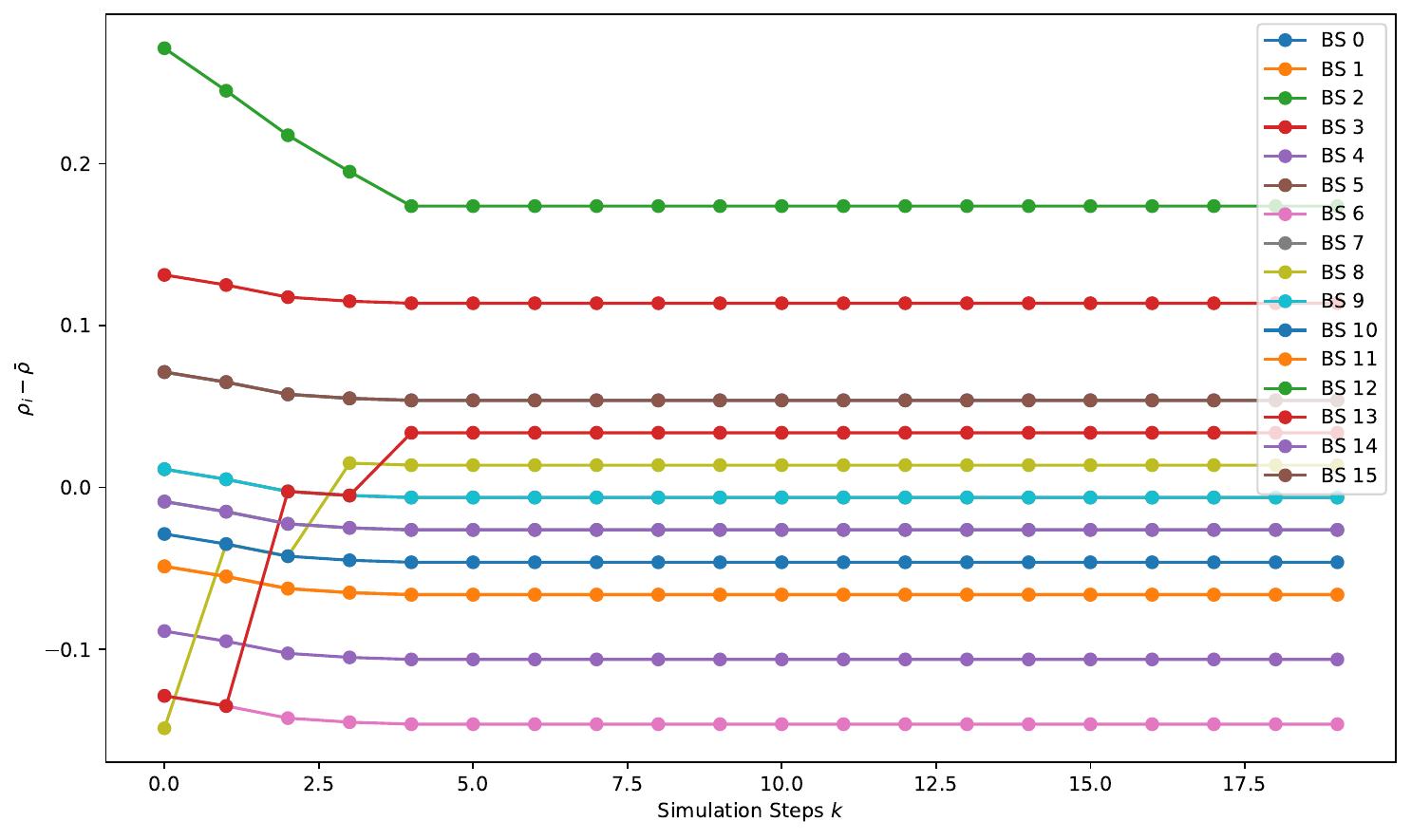}
  }
  \subfigure[]{
    \includegraphics[width=0.3\linewidth]{
    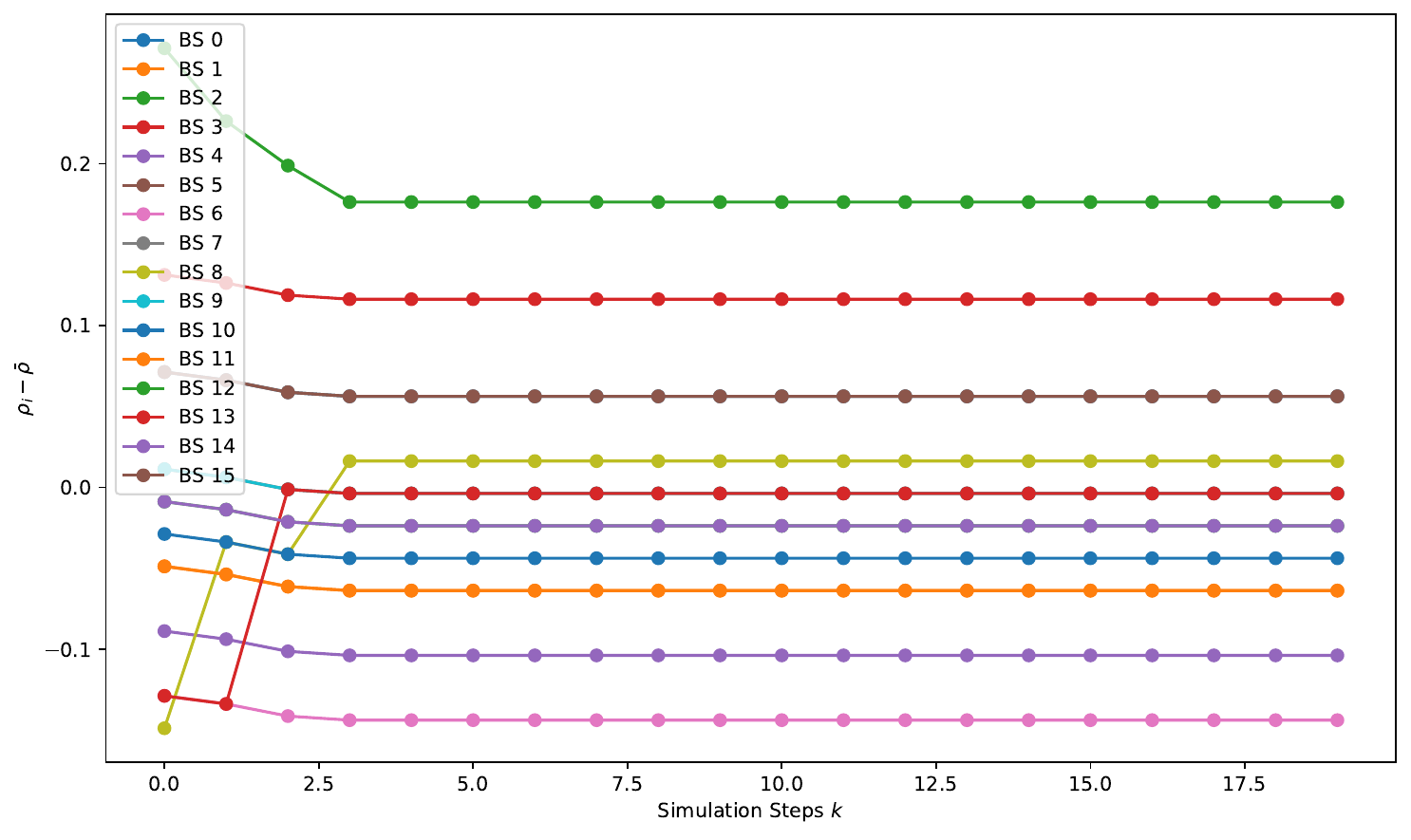}
  }

  \subfigure[]{
    \includegraphics[width=0.3\linewidth]{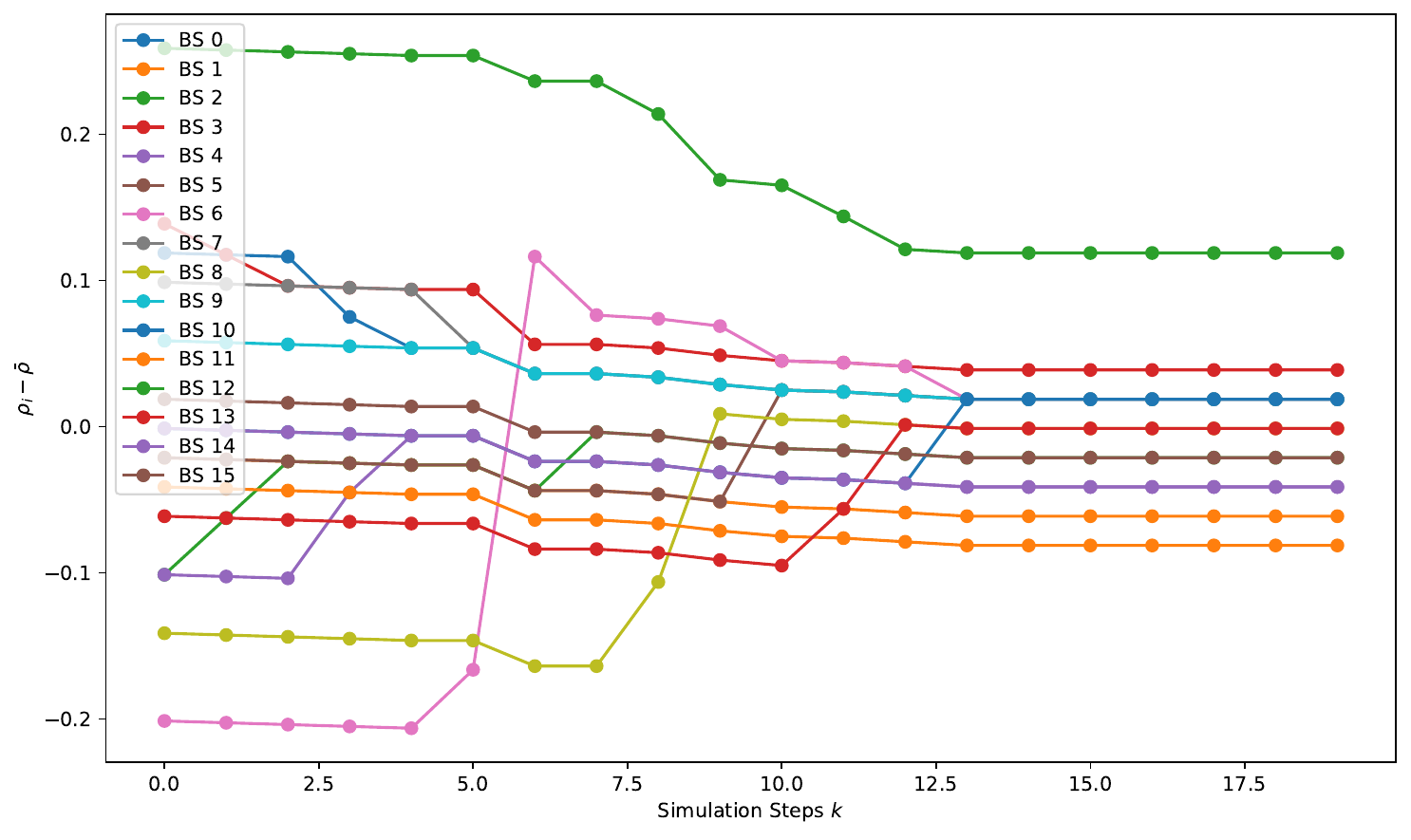}
  }
  \subfigure[]{
    \includegraphics[width=0.3\linewidth]{
    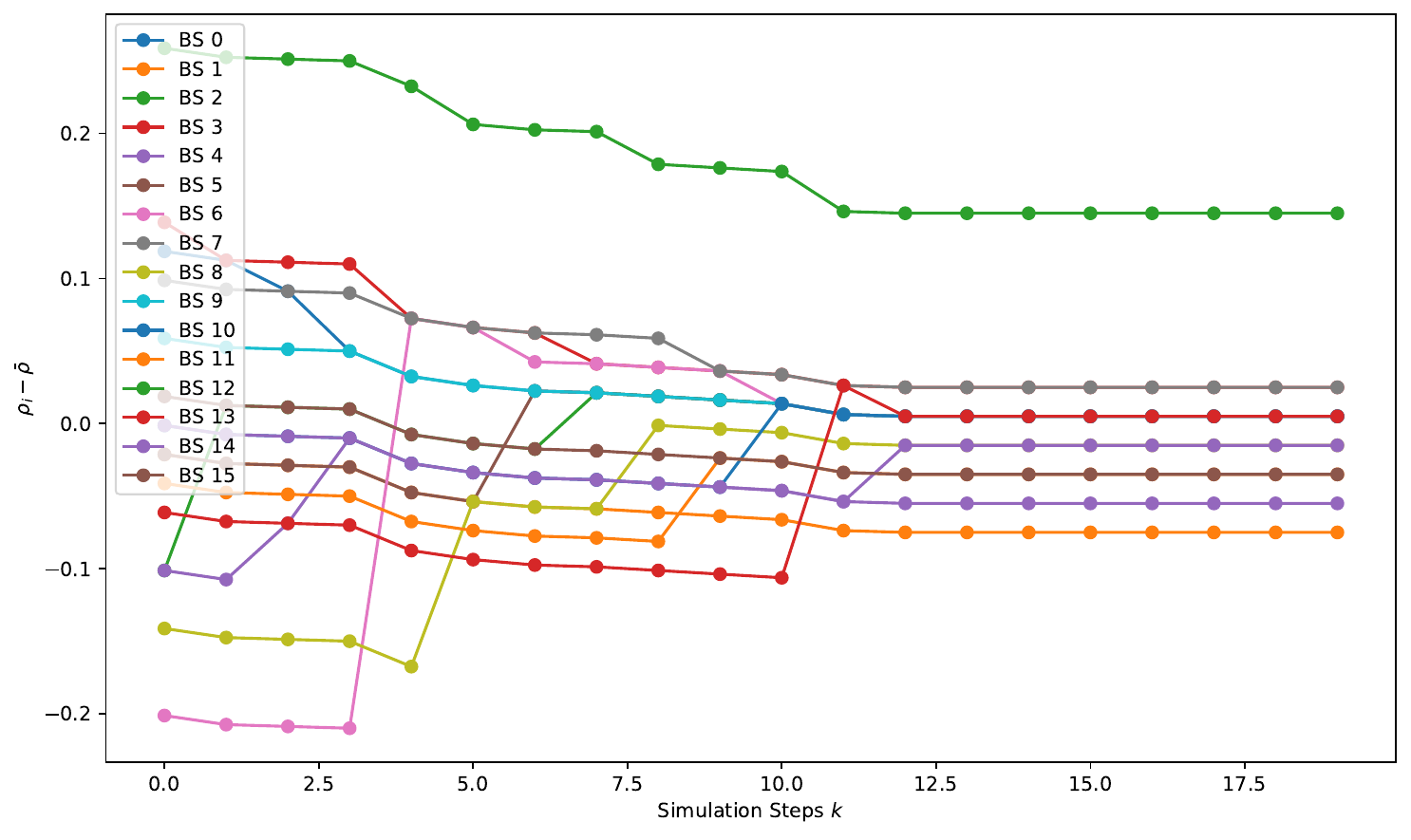}
  }
  \subfigure[]{
    \includegraphics[width=0.3\linewidth]{
    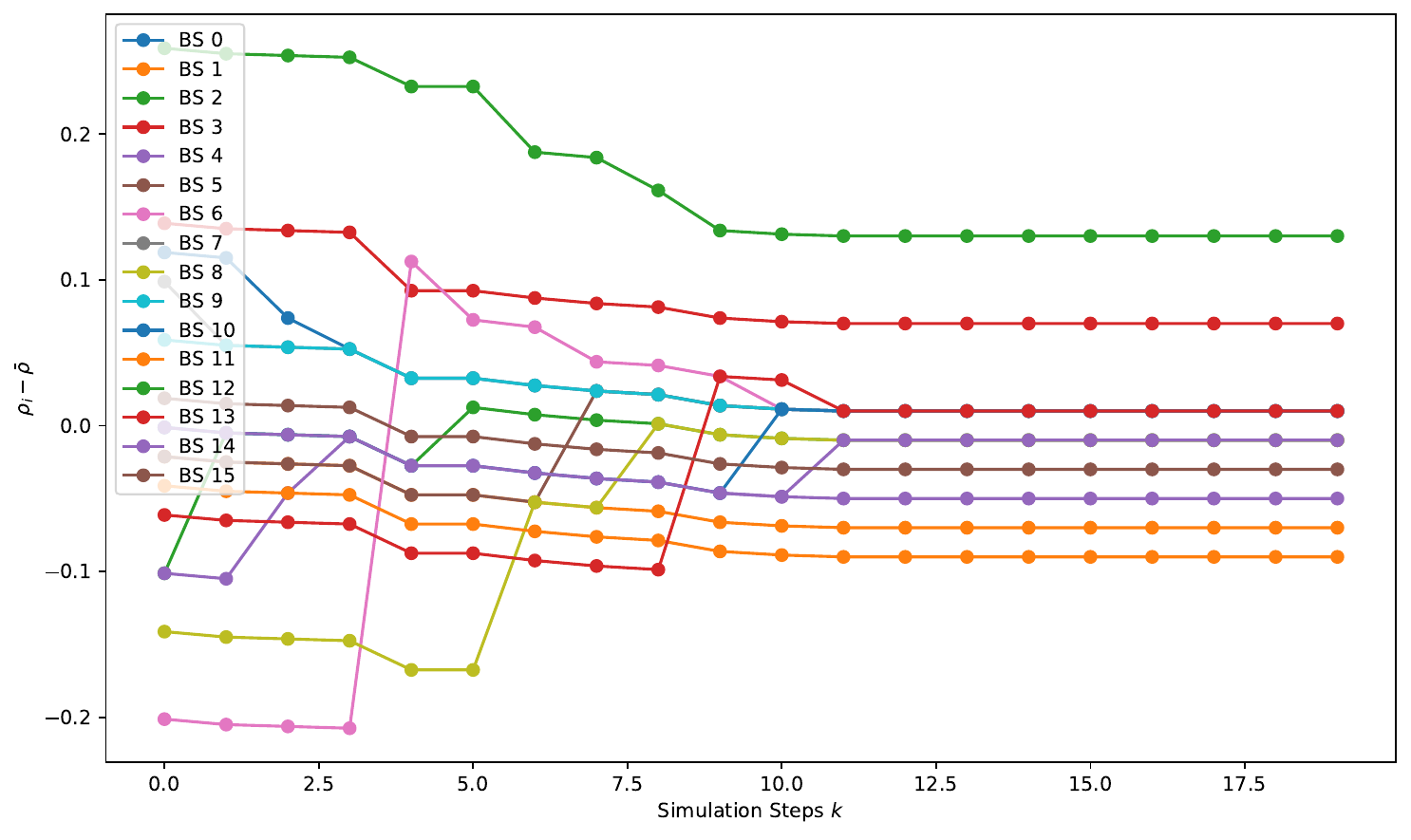}
  }
  
  \caption{Comparison of Load balance three different algorithms with different number of users. (a) (b) (c) with 200 users, (d) (e) (f) with 300 users, (g) (h) (i) with 400 users. (a) (d) (g) are generated by the first algorithm, (b) (e) (h) are generated by the second algorithm, and (c) (f) (i) are generated by the third algorithm.}
  \label{fig: comparison}
\end{figure*}

\begin{figure*}[ht]
  \centering
  \subfigure[]{
    \includegraphics[width=0.3\linewidth]{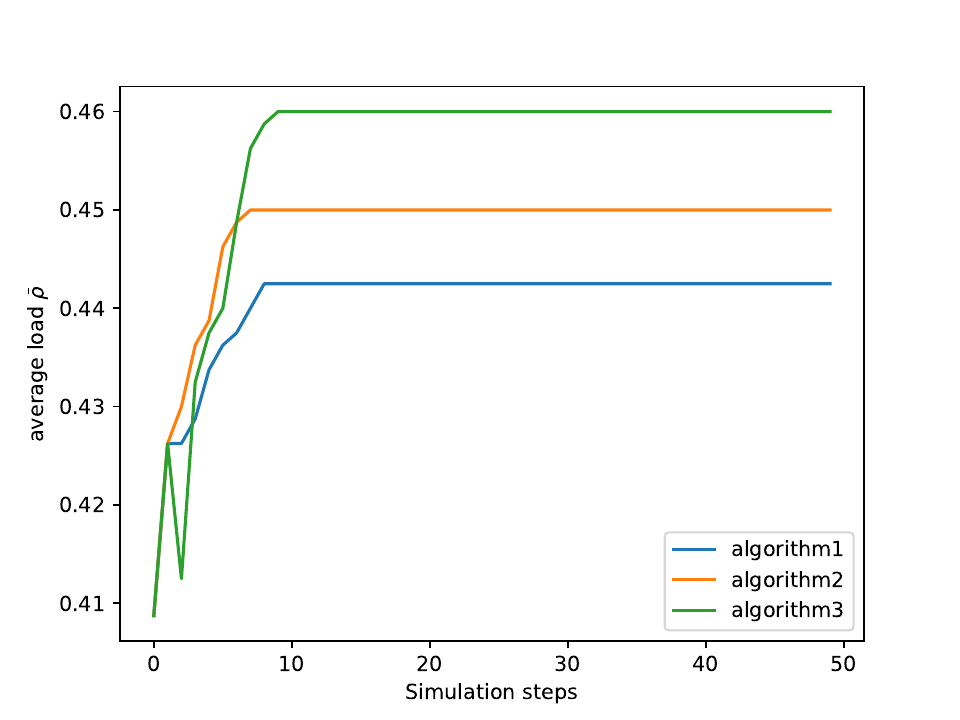}
    }
  \subfigure[]{
    \includegraphics[width=0.3\linewidth]{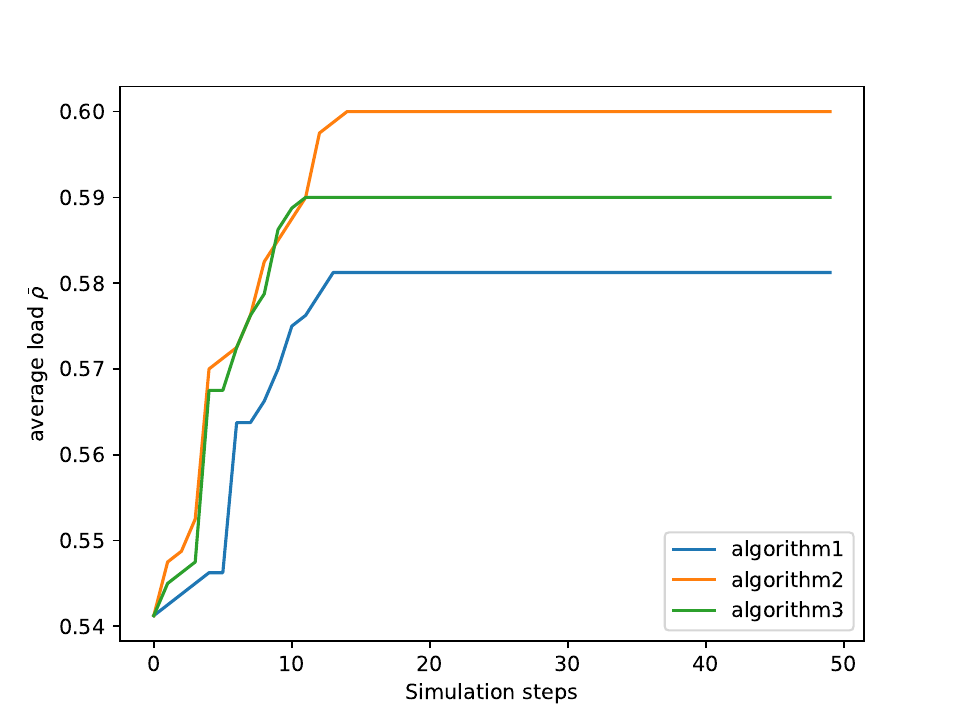}
   }
  \subfigure[]{
    \includegraphics[width=0.3\linewidth]{
    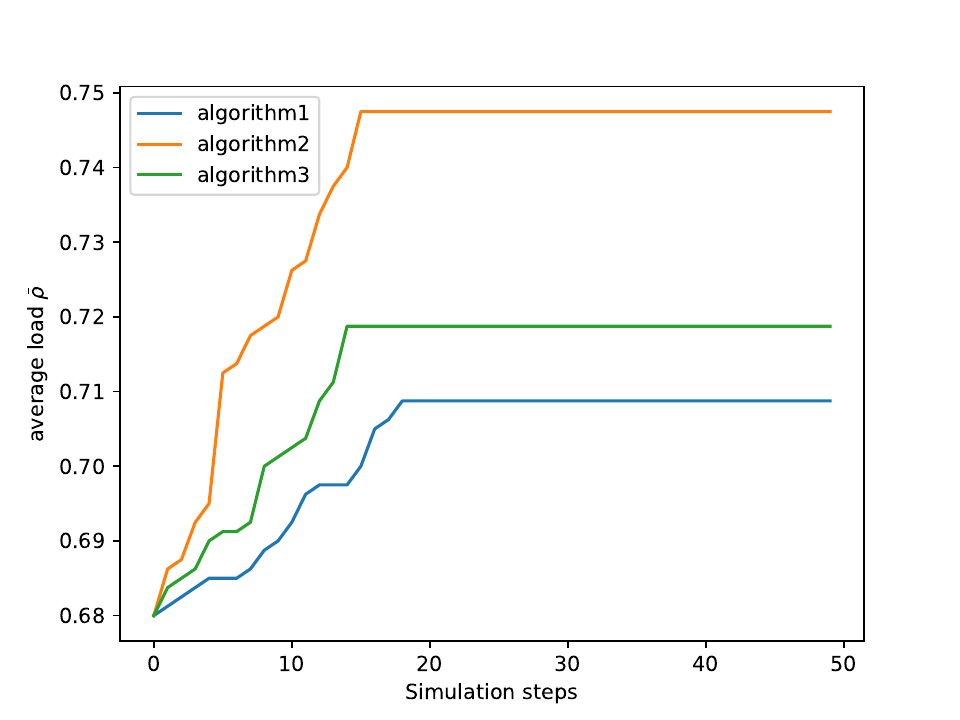}
  }
  \caption{Average load evolution with different numbers of users $300, 400, 500$ in (a) (b) and (c).}
  \label{fig: average_load_evolution}
\end{figure*}

\section{Results}
In this section, we validate our proposed analysis framework by simulation. We configure the system parameters according to \cite{park2017mobility, access2018evolved}. The bandwidth of a PRB is set as 180 KHz, each BS has 10 MHz of bandwidth. We set the inter-BS distance to 500 m. The transmit power of a BS is set to 46 dBm. The antenna gain of a BS is 14 dBi, and receiver gain is 5 dBi. The amount of the path loss from a BS to a user is modeled as $128.1+37.6 \log({\rm max}(d, 0.035))$, where $d$ is the distance between BS and user in km. The load balancing simulation experiment is conducted between 16 BSs shown in Fig.~\ref{fig: BSs distribution}.

In the first case, it shows what will happen if the CIO adaptively adjust only based on the load difference between neighbouring BSs to achieve load balance by a greedy algorithm. In Fig.~\ref{fig: average load}(a), the load difference between BSs gradually shrinks. However, loads of some BSs oscillate around a value. This means that some BSs keep handing over to each other without improving the performance of the system. Loads of BSs cannot synchronize to the same state. In Fig.~\ref{fig: average load}(b), we can see that the average loads of these BSs keep growing before 20 simulation steps. This is caused by the non-conservative load transfer between BSs. Even though after 20 simulation steps, the average load of the system remains the same, loads of some BSs still oscillate. This phenomenon validates our theory that loads of BSs cannot synchronise to the same state because of non-conservative load transfer. Therefore, instead of synchronising to the same state, the load balancing algorithm should aim to mitigate the load difference between BSs and stop handover if the load difference is smaller than a threshold.

Now we test the proposed handover algorithm with 3 different user selection policies. The first algorithm is shown in Algorithm~(\ref{algorithm: 1}). The second algorithm ranks the preference of handover user based on the load of users in the current BS, i.e. $\rho_{u, i}$ if user $u$ currently connects to BS $i$. The third algorithm ranks the preference of handover user based on the load of users in potential BS, i.e. $\rho_{u, j}$ if user $u$ service will be transferred from BS $i$ to $j$. These three algorithms are tested in three scenarios with different numbers of users. The load transfer accommodation is set as $|\eta_i(k)|\le 0.25*\sum A_{ij}|\rho_j(k)-\rho_i(k)|$ in these three scenarios. In the first case, the number of users is set as 200. The average load is $0.27125$. The threshold is set as $\rho_{\rm th}=0.5$. Only BS 12 is overloaded.
In the second case, number of users is 300. The average load is $0.40875$. The threshold is set as $\rho_{\rm th}=0.5$. BS $3$ and BS $12$ are overloaded.
In the third case, number of users is 400. The average load is $0.54125$. The threshold is set as 0.6. BS $0, 3, 7, 12$ are overloaded.
The comparison of different three algorithms which evolve with different simulation steps are shown in Fig.~(\ref{fig: comparison}). 
The initial average load and standard deviation (sd) is shown in Table~\ref{tab: 1}. The average load that evolves with algorithms in the scenario with $300, 400, 500$ users is shown in Fig.~(\ref{fig: average_load_evolution}). Since according to our analysis, loads of BSs cannot synchronise to the same state with the existence of non-conservative load transfer, the mechanism of load transfer accommodation should be set to prevent the endless handover to achieve the same state.  With the mechanism of load transfer accommodation, the endless handover has been eliminated (compare Fig.~\ref{fig: average load} and Fig.~\ref{fig: comparison}). From the comparison, it shows that the load balancing algorithm 1 is better than the other 2 algorithms in most situations in average load and standard deviation of load. Now we test the algorithms with different load transfer accommodations. The other settings are the same with the above experiment. The load transfer accommodation is set as $|\eta_i(k)|\le 0.3*\sum A_{ij}|\rho_j(k)-\rho_i(k)|$. The results are shown in Table~\ref{tab: 2}. And the load transfer accommodation is set as $|\eta_i(k)|\le 0.4*\sum A_{ij}|\rho_j(k)-\rho_i(k)|$. The results are shown in Table~\ref{tab: 3}. Algorithm 1 is still better than the other 2 algorithms in most situations. These experiments have demonstrated the effectiveness of our algorithm in evenly distributing the load with fewer standard deviations and an average value.  We can notice that in some special cases, algorithm 1 may not have the best performance, because the existence of the setting of load transfer accommodation. If the setting of the load transfer accommodation is too small, the handover number may be not enough to evenly distribute the load, which means the load still not converges to the bound. With the increase of the load transfer accommodation, the performance of algorithm 1 is better, while the other 2 algorithms' performance does not always improve.

\section{Conclusion}
We observe the phenomenon that the current load balancing algorithm effectively reduces the load difference among BSs, however, the system does not always fully converge to an ideal state and instead continues to oscillate. The oscillation may cause the ping-pong effect without improving the performance of the system.  To explain the potential
reasons behind this phenomenon and provide a pathway to improve the load balancing algorithm, we establish a networked dynamics model of load evolution in load balancing process, which is based on the load balancing mechanism and consistent with the existing load balancing algorithms. The dynamics model enables us to analyse the evolution behaviour of load dynamics in the long term.
Based on the established dynamics model, we established the stability criterion of the load balancing process and found that the existence of non-conservative load transfer causes the oscillation. The presence of the PRB scheduling mechanism within the base station can help mitigate this issue, allowing
the system to still converge to a desirable equilibrium. However, most current load balancing algorithms mainly consider the inter-BS handover algorithm. The condition for the existence of oscillation bound has been established as well as the estimation of this bound if the non-conservative load transfer exists and the stability cannot be guaranteed. We found that he bound of the oscillation is decided by the non-conservation error, eigenvalue spectral and the amount of transfer load. Based on this, we designed an appropriate strategy to alleviate the bound to improve the performance based on current load balancing algorithms. We validate our theory and the proposed algorithm based on the theory by simulations with different settings. These experiments have demonstrated the effectiveness of our algorithm in evenly distributing the load with fewer standard deviations and an average value. And the load transfer accommodation mechanism in the algorithm can effectively prevent the endless handover. Consider the impact of random walk of users on the load balancing process, in the future, we expect to establish a stochastic differential equation to analyse the evolution of load.

\bibliographystyle{IEEEtran}
\bibliography{ref}

\end{document}